\documentclass[journal]{IEEEtran}

\usepackage{cite}
\usepackage{subfigure}
\usepackage{latexsym}
\usepackage{amsmath}
\usepackage{amssymb}
\usepackage{amsfonts}
\usepackage{mathrsfs}
\usepackage[dvips]{graphicx}

\usepackage{epsfig}
\usepackage{color}
\usepackage{graphicx}

\newcommand{\be}{\begin{eqnarray}}
\newcommand{\ee}{\end{eqnarray}}
\newcommand{\beeq}{\begin{equation}}
\newcommand{\eeeq}{\end{equation}}
\newcommand{\beqs}{\begin{eqnarray*}}
\newcommand{\eeqs}{\end{eqnarray*}}

\newtheorem{thm}{Theorem}

\newtheorem{corol}[thm]{Corollary}
\newtheorem{prop}[thm]{Proposition}
\newtheorem{defi}[thm]{Definition}

\newtheorem{rem}[thm]{Remark}
\newtheorem{conj}[thm]{Conjecture}
\newcommand{\expec}{\mathbf{E}}

\begin{document}

\title{Feedback-Based Collaborative Secrecy Encoding over Binary Symmetric Channels}\label{chap5}
\author{George~T.~Amariucai, \emph{Member IEEE}~and~Shuangqing~Wei, \emph{Member IEEE}%
\thanks{G. Amariucai is with the Department of ECpE, Iowa State
University. E-mail: gamari@iastate.edu.}%
\thanks{S. Wei is with the Department of ECE, Louisiana State
University. E-mail: swei@ece.lsu.edu.}%
\thanks{This paper was  supported in part by the Board of Regents of Louisiana
under grants LEQSF(2004-08)-RD-A-17.}}
\maketitle

\begin{abstract}

In this paper we propose a feedback scheme for transmitting secret messages between two
legitimate parties, over an eavesdropped communication link. Relative to Wyner's traditional
encoding scheme \cite{wyner1}, our feedback-based encoding often yields larger rate-equivocation
regions and achievable secrecy rates. More importantly, by exploiting the channel
randomness inherent in the feedback channels, our scheme achieves a strictly positive
secrecy rate even when the eavesdropper's channel is less noisy than the legitimate receiver's channel.
All channels are modeled as binary and symmetric (BSC). We demonstrate the versatility of our feedback-based
encoding method by using it in three different configurations: the stand-alone configuration, the mixed configuration
(when it combines with Wyner's scheme \cite{wyner1}), and the reversed configuration. Depending on the channel
conditions, significant improvements over Wyner's secrecy capacity can be observed in all configurations.
\end{abstract}

\begin{IEEEkeywords}
 Eavesdropper Channel, Secrecy Capacity, Binary Symmetric Channels, Feedback.
\end{IEEEkeywords}


\section{Introduction}

In the context of a broadcast channel with confidential messages, it was
shown in \cite{csisz1} that a strictly positive secrecy capacity cannot
be achieved for any arbitrary pair of receiver/eavesdropper channels.
In particular, \cite{dijk1} proves that whenever the eavesdropper's
channel is \emph{less noisy} than the receiver's channel, no secret
messages can be exchanged between the legitimate transmitter and receiver
by the conventional method of \cite{wyner1}.

This motivated several works \cite{lai1}, \cite{maur1}, \cite{ahlsw1}, \cite{ardetsan}, \cite{gunduz} to
focus on alternative methods of achieving positive secrecy rates even when
the legitimate receiver has a worse channel than the eavesdropper. All
these works exploit the idea of  feedback channels.

The simple and interesting method of \cite{lai1} is based
on making the receiver jam the eavesdropper.  The receiver can
subtract its own jamming signal from the received signal, while the
wiretapper is kept totally ignorant of the confidential information flowing between
the legitimate users. The drawback of this approach is that the receiver
has to function in full duplex mode. Although an extension to half-duplex
mode is presented in \cite{lai1} for binary symmetric channels, it relies on the
assumption that the transmission of symbol $0$ is equivalent to the absence
of a  physical signal.  We believe that under this assumption, the binary
symmetric channel is no longer valid as a simplified model for a physical wireless
channel.

More recently, \cite{ardetsan}, \cite{gunduz} both use a secret key
to enhance the secrecy throughput of Wyner's scheme. In \cite{ardetsan} the secret key is
communicated through an error-free secure channel, while in \cite{gunduz} it is transmitted
using Wyner's scheme on the feedback channels (and thus its secrecy is subject to Alice's
feedback channel being better than Eve's). An interesting idea of \cite{gunduz} is to use
time-sharing on the feedback link. Part of the feedback transmission is used to generate the
secret key, while the remaining part is used to transmit random symbols with the
purpose of providing the ``common randomness'' necessary for our secrecy encoding scheme described
in this paper. A mixed secrecy encoding strategy inspired by \cite{myself6} is proposed in
\cite{gunduz}. The main idea behind this strategy is to simultaneously transmit a combination
of secret messages, encoded by different methods.  However, while a message encrypted by a secret
key can be transmitted at the same time as a secret message encoded by Wyner's scheme, the
additional secret message encrypted with the use of a random feedback sequence does not maintain secrecy.
The exact reasons why both Section IV.B. of \cite{myself6} and the proposed schemes of \cite{gunduz} are
incorrect will be revealed in Section \ref{5sect5}. None of the previously mentioned works considers
the impact of feedback transmission
on the overall bandwidth use. This drawback becomes critical in \cite{gunduz}, where it results in
the ``secrecy rates'' bearing no physical meaning, as will be shown in Appendix \ref{5appgunduzwrong}.

The concept of \emph{common randomness} is introduced in \cite{maur1, ahlsw1}.
Such randomness can be acquired if all terminals attempt to decode (note that a necessary
condition is that the eavesdropper cannot decode perfectly) a sequence of random bits,
as for example a data stream transmitted by a satellite at very low signal to noise ratio (SNR)
\cite{maur1}.  Both \cite{maur1} and \cite{ahlsw1} study the case when
the legitimate users agree on a secret key by employing  repetitive
protocols, which are not efficient for regular data transmission.

The idea developed in this paper is inspired by a particular case in
\cite{maur1}.  As an example and motivation for the feedback approach
to secrecy in the classical Alice (transmitter) - Bob (receiver) - Eve
(eavesdropper) scenario, \cite{maur1} develops a scheme where the common
randomness is not received from some external source (like a satellite),
but introduced by Alice herself, and functions as a secret key which
allows Bob to share a secret message with Alice over a public, error
free channel.  Our model changes the roles of Alice and Bob.  Although at
some point we make use of the same concept of public error free channel,
we provide techniques that create such a channel, and show how these techniques impact the overall
secrecy rate. Our results explicitly count the loss in the total rate due to
the transmission of feedback.

While sharing functional similarities with the well-known \emph{one-time
pad} \cite{schneier} encryption scheme, our approach is radically
different in that it requires no secret key to be shared by the legitimate
parties before the initiation of the transmission protocol (except maybe a small
secret key that guarantees authenticity as in \cite{maur1}). Instead it
exploits the channel randomness as means of confusing the eavesdropper.

Our contributions  can be summarized as follows: 
\begin{itemize}

\item  We show how an adaptation of Maurer's scheme \cite{maur1} can be used to achieve a
non-empty rate-equivocation region and hence a strictly positive secrecy rate over
binary symmetric channels (BSCs) even
when the forward channel between Alice and Eve is {\em less noisy} than
the forward channel between Alice and Bob, regardless of the feedback channel
quality between Bob and Alice or Bob and Eve.

\item Our results also indicate how  the forward channel capacities
scale the overall secrecy rate and what penalties are incurred by the transmission
of feedback sequences.

\item We show that even if the forward channel from Alice to Bob is less
noisy than the channel from Alice to Eve, feedback can sometimes further improve
the achievable rate-equivocation region obtained using Wyner's classical
method \cite{wyner1}. This is done by dividing the transmission over the forward
channel into two parts, as in \cite{csisz1}. Thus, we transmit a secret message at a rate
less than the secrecy capacity \cite{wyner1}, and allow room for an additional common message,
which carries information ``encrypted'' with the help of the feedback sequence.
The optimal way of splitting the forward message rate is found numerically.

\item We prove that, for a two-user broadcast channel with both channels binary and
symmetric, the optimal auxiliary random variable of \cite{csisz1} needed to encode both a secret and
a common message into the transmitted sequence has an alphabet of size not more than
three. Moreover, we conjecture that the optimal alphabet is binary. If the auxiliary random
variable is considered to be binary (whether or not this results in loss of optimality),
we prove that the optimal auxiliary channel \cite{csisz1} that links it
to the input of the physical channel is binary and symmetric.

\item Finally, we take our scheme a step further and implement it on the reverse channel
(from Bob to Alice, rather than from Alice to Bob), in order to generate a secret key.
Alice uses this key as a one-time pad for the transmission of a secret message.
\end{itemize}

\begin{figure*}[!ht]
\centering
\includegraphics[scale=0.6]{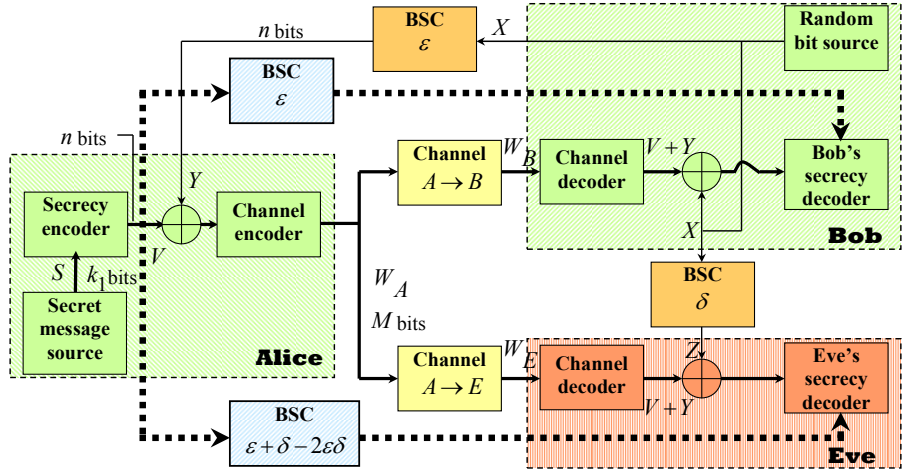}
\caption{System model.}\label{5channelmodelwt}
\end{figure*}

The sequel is organized into seven sections. 
Sections \ref{5sect2} and \ref{5sect3} describe the kernel of our scheme.
Our adaptation of Maurer's idea \cite{maur1}, including the channel model and the
transmission protocol are presented in Section \ref{5sect2} under the assumption that
the forward channels are error free. The public error free channel and the overall
rate-equivocation region are discussed in Section \ref{5sect3} for a general value of the forwarding
rate. Section \ref{5sect4} deals with the special case when the eavesdropper's
forward channel is less noisy than the legitimate receiver's forward channel, while
section \ref{5sect5} extends the model to the case when the eavesdropper's forward channel is worse
than the legitimate receiver's. An alternative scheme, which reverses our protocol to generate a secret key,
is provided in Section \ref{5sect6}. Finally, conclusions are drawn in Section \ref{5sect8}.


\section{The Kernel}

\subsection{The Unscaled Rates}\label{5sect2}

Consider the classical Alice (transmitter) - Bob
(receiver) - Eve (eavesdropper) scenario with binary symmetric channels
(BSCs) between any pair of users. We assume that Eve's only form of interfering with
the transmission is eavesdropping. Although our present treatment is
restrictive to binary channels, the principles and results therein can
be easily extended to more complex models.

The proposed model is depicted in Figure \ref{5channelmodelwt}. The
transmitter (Alice) wants to communicate the outputs of a source
$\mathscr{S}$ of entropy $H_s$ to the legitimate
receiver (Bob), and maintain some level of secrecy towards the wiretapper
(Eve).  The channel $A\to B$ from Alice to Bob is a BSC characterized
by its crossover probability $\epsilon_f$, while the binary symmetric
channel $A\to E$ from Alice to Eve is characterized by the crossover
probability $\delta_f$.  Similarly, the feedback BSCs $B\to A$ (Bob to
Alice) and $B\to E$ (Bob to Eve) are characterized by their crossover
probabilities $\epsilon_b$ and $\delta_b$, respectively.

The transmission protocol associated with the channel model in Figure
\ref{5channelmodelwt} is an adaptation of Maurer's scheme \cite{maur1}
and is described as follows. Bob feeds back a sequence
$\mathbf{x}$ of $n$ bits representing the independent realizations
of a Bernoulli random variable $X$ with expectation $\expec[X]=0.5$.
Since the bits are independent and identically distributed (i.i.d),
Alice's and Eve's estimate of each bit should be based solely on the
corresponding received bit.  Therefore, the bit error probabilities
that affect Alice's and Eve's decoding are $\epsilon_b$ and $\delta_b$
respectively. Denote the feedback sequences received by Alice and Eve
as $\mathbf{y}$ and $\mathbf{z}$, respectively.

At this point, our feedback-based protocol assumes that Alice can share information with
both Bob and Eve through an error free public channel, just like in \cite{maur1}.
The implications of
achieving such an error free channel are discussed in Section \ref{5sect3}.

Since an error free public channel cannot protect Alice's information
from the eavesdropper Eve, the protocol has to artificially create a pair
of channels that are adequate for the transmission of secret messages.

For this purpose, if Alice needs to send an n-dimensional sequence $\mathbf{v}$ to Bob,
she first computes $\mathbf{v}\oplus\mathbf{y}$, where $\oplus$ denotes
addition $\mod 2$, and feeds this signal through the error free channel.
Since $\mathbf{x}$ is a sequence of i.i.d. symbols with a uniform distribution over $\{0,1\}$,
the same property holds for the BSC output $\mathbf{y}$ and,
by the \emph{Crypto lemma}\footnote{Special care should be applied when using the Crypto lemma \cite{erez2}. For instance,
if $\mathcal{C}$ is a compact Abelian group and $X$ and $E$ are random variables over $\mathcal{C}$ such that
$X$ is \emph{independent of $E$} and uniformly distributed over $\mathcal{C}$, then $X+E$ is uniform and \emph{independent of $E$}.
However, $E$ is not independent of $(X,X+E)$.} \cite{erez2}, \cite{lai1}, for $\mathbf{v}\oplus\mathbf{y}$.

Both Bob and Eve receive $\mathbf{v}\oplus\mathbf{y}$ with no errors.
In order to obtain the original message $\mathbf{v}$, the optimal strategy for Bob is to compute
$\mathbf{v}\oplus\mathbf{y}\oplus \mathbf{x}$, while Eve's best strategy is to compute $\mathbf{v}\oplus\mathbf{y}
\oplus \mathbf{z}$ \cite{maur1}.

As a consequence, a bit error probability of $\epsilon_B=\epsilon_b$ will affect Bob's estimate of $\mathbf{v}$, while a bit error
probability of $ \epsilon_E=\epsilon_b+\delta_b-2\epsilon_b\delta_b$ will affect Eve's estimate \cite{maur1}.
The result is an equivalent system in which Eve's channel is a degraded version of Bob's channel, and
which is therefore adequate for the transmission of secret messages from Alice to Bob.
In other words, standard secrecy encoding can be performed for this equivalent system so that 
the $n$-sequence $\mathbf{v}$ carries a secret message $\mathbf{s}^{k_1}$ (which will hence forth be represented as a sequence of
$k_1$ source symbols). A total transmission rate arbitrarily close to 
\be 
R_{t,u}=1-h(\epsilon_b)
\ee
can be achieved as $n\to\infty$, where $h(\cdot)$ represents the binary entropy function
$h(x)=-x\log_2(x)-(1-x)\log_2(1-x)$.

We shall now restate some of the definitions in \cite{wyner1}
and then show how Theorem 2 of \cite{wyner1} can be readily applied to our scenario.

\begin{defi}
The equivocation of the source $\mathscr{S}$ of entropy $H_s$ at Eve is defined as:
\be
\Delta=\frac{1}{k}H(\mathbf{s}^{k}|\mathbf{w_E^M}),
\ee
where the sequence $\mathbf{s}^k$ of $k$ source symbols are encoded
into a codeword $\mathbf{w_A^M}$ of length $M$ which is transmitted over the broadcast channel, and 
received by Eve as $\mathbf{w_E^M}$.
\end{defi}

\begin{defi}\label{5defwyner}
The rate-equivocation pair $(R,d)$ is achievable if for any $\nu>0$ there exists an $(M,k,\Delta,\overline{P_e})$ code as defined in
\cite{wyner1} such that:
\be
\begin{array}{lcr}
\frac{kH_s}{M}\geq R-\nu,
&
\Delta\geq d-\nu,
&
\overline{P_e}\leq \nu
\end{array}
\ee
where $\overline{P_e}$ is the average error probability in decoding for $\mathbf{s}^{k}$ at Bob.
\end{defi}

\begin{thm} \label{5theorem2wyner} (\emph{Theorem 2 from \cite{wyner1}})
A rate-equivocation pair $(R,d)$ is achievable for Wyner's scheme with 
\emph{discrete memoryless symmetric channels} if and only if
\be
\begin{array}{lcr}
0\leq R\leq C_M,&
0\leq d\leq H_s,&
Rd\leq H_s C_s,
\end{array}
\ee
where $C_s=C_M-C_{MW}$ is the secrecy capacity 
(representing the maximum rate at which the outputs of the source $\mathscr{S}$ can be conveyed from Alice to Bob,
while remaining perfectly secret to Eve) achievable by Wyner's scheme in this case, $C_M$ is the capacity of Bob's channel,
and $C_{MW}$ is the capacity of Eve's channel.
\end{thm}

The following corollary, which will prove useful in the sequel,
follows directly from Theorem \ref{5theorem2wyner} and Definition \ref{5defwyner}.

\begin{corol}\label{5cor1wyner}
If $(R,d)$ is an achievable rate-equivocation pair, then as $M\to\infty$ the number of secret source symbols $k$ that can be encoded
into the $M$-sequence $\mathbf{w_A^M}$ can approach the upper-bound
\be
k_u=\frac{MC_s}{d}.
\ee
\end{corol}
\begin{proof}
By Theorem \ref{5theorem2wyner} and Definition \ref{5defwyner} we have
$\frac{kH_s}{M}\geq R-\nu$, which implies $\frac{kd}{M}\geq \frac{1}{H_s}(Rd-\nu d)$, and
taking the limit $Rd=H_s C_s$ we get $\frac{kd}{M}\geq C_s-\frac{\nu d}{H_s}$.
But according to Theorem \ref{5theorem2wyner}, we have $d\leq H_s$, hence, as $\nu\to 0$, if we pick a large
enough $M$ we can obtain $\frac{kd}{M}\to C_s$, or equivalently $k \to \frac{MC_s}{d}$.
\end{proof}

If we apply Theorem \ref{5theorem2wyner} to the pair of equivalent channels derived above, we can conclude that
there exists a $(n,k_1,\Delta_1,\overline{P_{e,1}})$ code satisfying
$\frac{k_1H_s}{n}\geq R-\nu$,
$\Delta_1\geq d-\nu$,
and $\overline{P_{e,1}}\leq \nu$
if and only if
$0\leq R\leq R_{t,u}$,
$0\leq d\leq H_s$,
$Rd\leq H_s R_{s,u}$,
where $R_{s,u}$ is the maximum achievable \emph{secrecy rate} of \cite{wyner1, csisz1}:
\be \label{5Rsu}
R_{s,u}=h(\epsilon_b+\delta_b-2\epsilon_b\delta_b)-h(\epsilon_b).
\ee

Several comments are in order.
First, note that $R_{s,u}>0$ -- and therefore the rate-equivocation region as defined in
\cite{wyner1} is non-empty -- unless $\delta_b\in\{0,1\}$ (the assumption that feedback channels exist
implies $\epsilon_b \neq 0.5$)

Second, the rates $R_{t,u}$ and $R_{s,u}$ do not represent the \emph{overall} transmission and secrecy rates of our model,
since a pair of binary symmetric channels such as the forward $A\to B$ and $A\to E$ channels cannot provide error free
transmission at infinite rates. The information encoded in the sequence $\mathbf{v}$ mentioned above has to
pass through one of these channels in order to be available at the other two terminals.
While this ``correction'' will be considered in Section \ref{5sect3}, we shall denote the rates $R_{t,u}$ and $R_{s,u}$ as
\emph{the unscaled transmission and secrecy rates}, respectively.

Third, note that under the above protocol, an independent feedback sequence $\mathbf{x}$
is transmitted every time for each new information-carrying sequence $\mathbf{v}$.
Eve's resulting error sequence is always different and independent, and acts like a \emph{one-time pad}
\cite{schneier}. As is the case with a one-time pad, the feedback sequence cannot be recycled.
If only one feedback sequence is transmitted and used for a set of several messages, Eve's equivocation about
the whole set will be the same as her equivocation about any one message in the set.

Therefore, an additional rate penalty has to be introduced to address the channel uses required for the
feedback of $\mathbf{x}$, as will be shown in Section \ref{5sect3}.


\subsection{The Overall Rate-Equivocation Region and Secrecy Rate}\label{5sect3}

This section shows how the overall transmission rates of our model depend on the \emph{unscaled}
rates of the equivalent system presented in Section \ref{5sect2} and on the transmission rates used over the
forward binary symmetric channels.

In Section \ref{5sect2} we showed that, if feedback is allowed, we can artificially form an equivalent system that
allows encoding by Wyner's scheme \cite{wyner1}. All that is needed is an error free
public channel to support the transmission of the $n$-sequence $\mathbf{v}\oplus\mathbf{y}$.
By the channel coding theorem, this channel is readily available if $\mathbf{v}\oplus\mathbf{y}$
is transmitted at a rate $R_{AB,fb}$ (the notation stands for the rate at which the feedback
processed signal is transmitted from Alice to Bob) less than the capacity of the $A\to B$
channel $C_{AB}=1-h(\epsilon_f)$.

\begin{prop}\label{prop1IIB}
There exists a channel code $(M,n,P_{e,c})$ (where $n$ is the size of the message, $M$ is the size of the codeword and
$P_{e,c}$ is the code's average error probability) that can transport the sequence $\mathbf{v}\oplus\mathbf{y}$
over the forward channel in such a manner that the secret message $\mathbf{s}^{k_1}$ is recovered
with asymptotically no errors by Bob.
\end{prop}

\begin{proof}
Denote the error sequences introduced by the feedback channels by
$\mathbf{e_{bA}}$ -- for Alice -- and $\mathbf{e_{bE}}$ -- for Eve.
According to \cite{wyner1} if the rate of the secret message is less than $R_{s,u}$, then
there exists an encoding/decoding technique such that for any $\nu>0$ there exists $N_0>0$ such that
the average probability of correctly decoding for the secret message $\mathbf{s}^{k_1}$ is
\be\label{5wynercodingthm}
\sum_{\mathbf{s}}Pr\{\mathbf{s}^{k_1}\}\sum_{\mathbf{v},\mathbf{e_{bA}}}Pr\{\mathbf{e_{bA}}\}Pr\{\mathbf{v}|\mathbf{s}^{k_1}\}
\cdot{}\nonumber\\ \cdot Pr\{\psi (\mathbf{v}\oplus\mathbf{e_{bA}})=\mathbf{s}^{k_1}\}\geq 1-\nu
\ee
for $n>N_0$, where $\psi (\cdot )$ is Bob's secrecy decoder.

Moreover, according to Gallager's second corollary of Theorem 5.6.2. \cite{gallager}, there exists a $(M,n,P_{e,c})$ code
for Bob's forward channel with the property that if the transmission rate is $\frac{n}{M}=R_{AB,fb}<C_{AB}$, then 
for any $\nu>0$ there exists $N_1>0$ such that the average probability of correctly decoding a given
transmitted message $\mathbf{t}$ is
\be\label{5gallagercodingthm}
1-P_{e,c}=\sum_{\mathbf{w_B},\mathbf{t}}Pr\{ \mathbf{t} \}Pr\{\mathbf{w_B}|\mathbf{t}\}Pr\{\phi (\mathbf{w_B})=\mathbf{t}\}\geq 1-\nu
\ee
for $n>N_1$, where $\phi (\cdot )$ is Bob's channel decoder and $\mathbf{w_B}$ is Bob's received sequence over
the forward channel (when $\mathbf{w_A}$ is transmitted by Alice).
Note that our decoding method consists of separate channel and secrecy decoding. That is, Bob estimates the secret message
$s$, as $\widehat{s}=\psi(\phi(\mathbf{w_B})\oplus \mathbf{x})$. There is no guarantee that this separate decoding method is optimal.
We define Bob's optimal (joint) decoder $\xi(\cdot)$, yielding the optimal estimate $\widetilde{s}=\xi(w_B)$.
Given the feedback sequence $\mathbf{x}$, we can lower bound 
\be\label{5problowbound}
Pr\{\xi (\mathbf{w_B})=\mathbf{s}^{k_1}\}\geq{} \nonumber\\
{}\geq \sum_{\mathbf{t}} Pr\{\phi (\mathbf{w_B})=\mathbf{t}\}Pr\{\psi (\mathbf{t\oplus \mathbf{x}})=\mathbf{s}^{k_1}\}.
\ee

Thus given the feedback sequence $\mathbf{x}$,
Bob's average probability of correctly decoding for the secret message $\mathbf{s}^{k_1}$ can be lower bounded as
\be
\sum_{\mathbf{s}^{k_1}}Pr\{\mathbf{s}^{k_1}\}\sum_{\mathbf{v},\mathbf{e_{bA}}}
Pr\{\mathbf{e_{bA}}\}Pr\{\mathbf{v}|\mathbf{s}^{k_1}\}\sum_{\mathbf{x}}Pr\{\mathbf{x}\}\cdot{}\nonumber\\
{}\cdot \sum_{\mathbf{w_B}}  Pr\{\mathbf{w_B}|\mathbf{v}\oplus\mathbf{e_{bA}}\oplus\mathbf{x}\}Pr\{\xi (\mathbf{w_B})
=\mathbf{s}^{k_1}\}\stackrel{(a)}{\geq} {}\nonumber\\ 
{}\geq \sum_{\mathbf{s}^{k_1}}Pr\{\mathbf{s}^{k_1}\}\sum_{\mathbf{v},\mathbf{e_{bA}}}
Pr\{\mathbf{e_{bA}}\}Pr\{\mathbf{v}|\mathbf{s}^{k_1}\}\sum_{\mathbf{x}}Pr\{\mathbf{x}\}\cdot{}\nonumber\\
{}\cdot \sum_{\mathbf{w_B}}  Pr\{\mathbf{w_B}|\mathbf{v}\oplus\mathbf{e_{bA}}\oplus\mathbf{x}\}
\sum_{\mathbf{t}} Pr\{\phi (\mathbf{w_B})=\mathbf{t}\}\cdot{}\nonumber\\
{}\cdot Pr\{\psi (\mathbf{t\oplus \mathbf{x}})=\mathbf{s}^{k_1}\}\stackrel{(b)}{\geq} {}\nonumber\\ 
{}\geq \sum_{\mathbf{s}^{k_1}}Pr\{\mathbf{s}^{k_1}\}\sum_{\mathbf{v},\mathbf{e_{bA}}}
Pr\{\mathbf{e_{bA}}\}Pr\{\mathbf{v}|\mathbf{s}^{k_1}\}\sum_{\mathbf{x}}Pr\{\mathbf{x}\}\cdot{}\nonumber\\
{}\cdot \sum_{\mathbf{w_B}}  Pr\{\mathbf{w_B}|\mathbf{v}\oplus\mathbf{e_{bA}}\oplus\mathbf{x}\}
\cdot{}\nonumber\\
{}\cdot Pr\{\phi (\mathbf{w_B})=\mathbf{v}\oplus\mathbf{e_{bA}}\oplus\mathbf{x}\}\cdot{}\nonumber\\
{}\cdot Pr\{\psi (\mathbf{\mathbf{v}\oplus\mathbf{e_{bA}}})\}
=\mathbf{s}^{k_1}\}\stackrel{(c)}{=}{}\nonumber\\ 
{}= \sum_{\mathbf{s}^{k_1}}Pr\{\mathbf{s}^{k_1}\}\sum_{\mathbf{v},\mathbf{e_{bA}}}
Pr\{\mathbf{e_{bA}}\}Pr\{\mathbf{v}|\mathbf{s}^{k_1}\}\cdot{}\nonumber\\
{}\cdot Pr\{\psi (\mathbf{\mathbf{v}\oplus\mathbf{e_{bA}}})\}
=\mathbf{s}^{k_1}\} \sum_{\mathbf{x}}Pr\{\mathbf{x}\}\sum_{\mathbf{w_B}} {}\nonumber\\
{} Pr\{\mathbf{w_B}|\mathbf{v}\oplus\mathbf{e_{bA}}\oplus\mathbf{x}\}\cdot{}\nonumber\\
{}\cdot Pr\{\phi (\mathbf{w_B})=\mathbf{v}\oplus\mathbf{e_{bA}}\oplus\mathbf{x}\}\stackrel{(d)}{\geq} {}\nonumber\\ 
{}\geq (1-\nu)\sum_{\mathbf{s}^{k_1}}Pr\{\mathbf{s}^{k_1}\}\sum_{\mathbf{v},\mathbf{e_{bA}}}
Pr\{\mathbf{e_{bA}}\}Pr\{\mathbf{v}|\mathbf{s}^{k_1}\}\cdot{}\nonumber\\
{}\cdot Pr\{\psi (\mathbf{\mathbf{v}\oplus\mathbf{e_{bA}}})\}
=\mathbf{s}^{k_1}\}\stackrel{(e)}{\geq} (1-\nu)^2.
\ee
Inequality $(a)$ follows from (\ref{5problowbound}), inequality $(b)$ from the fact that
$\sum_{t} F(\mathbf{t})\geq F(\mathbf{t})|_{\mathbf{t}=\mathbf{v}\oplus\mathbf{e_{bA}}\oplus\mathbf{x}}$ for
any positive function $F$, while the equality $(c)$ from simply re-arranging the terms.
In inequality $(d)$ we used (\ref{5gallagercodingthm})
and the fact that $ Pr\{\mathbf{v}\oplus\mathbf{e_{bA}}\oplus\mathbf{x}\}=Pr\{\mathbf{x}\}$ (due to the Crypto
lemma \cite{erez2}), while inequality $(e)$ follows directly from (\ref{5wynercodingthm}).
The resulting average error probability at Bob is thus
\be
\overline{P_e}<2\nu-\nu^2,
\ee
which goes to zero as $\nu\to 0$.
\end{proof}

Denote $C_{AE}=1-h(\delta_f)$ the capacity of Eve's forward channel.
Note that if $C_{AE}\geq C_{AB}$, Eve will also be able to decode the sequence $\mathbf{v}\oplus\mathbf{y}$ with
no errors asymptotically. However, Eve's equivocation about the secret message $\mathbf{s}^{k_1}$ is maintained due to the feedback
processing. On the other hand, if $C_{AE}<C_{AB}$, Eve cannot decode for the message $\mathbf{v}\oplus\mathbf{y}$.
Under this scenario, a secret message can be transmitted from Alice to Bob by Wyner's scheme, without using any
feedback. The optimal tradeoff between the rate of encoding a secret message directly through Wyner's scheme and
the rate $R_{AB,fb}$ at which a feedback-processed secret message should be forwarded to Bob will be discussed in
Section \ref{5sect5}. In what follows, we prove that Eve's equivocation about the feedback-processed secret message
$\mathbf{s}^{k_1}$ is maintained regardless of the forwarding rate $R_{AB,fb}$.

\begin{prop}\label{prop2IIB}
Eve's equivocation about the secret message does not decrease because of channel coding for the forward channel.
\end{prop}

\begin{proof}
Let $\mathbf{w_E}$ denote Eve's received signal over the forward channel
and $\mathbf{s}^{k_1}$ denote the secret message.
Also, recall the error sequences corresponding to the feedback channels were denoted by
$\mathbf{e_{bA}}$ (for Alice's feedback channel) and $\mathbf{e_{bE}}$ (for Eve's feedback channel).

Eve's equivocation about the secret message is
\be
H(\mathbf{s}^{k_1}|\mathbf{w_E},\mathbf{x}\oplus\mathbf{e_{bE}}) \geq
H(\mathbf{s}^{k_1}|\mathbf{v}\oplus\mathbf{y},\mathbf{x}\oplus\mathbf{e_{bE}})={}\nonumber\\
{}=H(\mathbf{s}^{k_1}|\mathbf{v}\oplus\mathbf{e_{bE}}\oplus\mathbf{e_{bA}},\mathbf{x}\oplus\mathbf{e_{bE}})={}\nonumber\\
=H(\mathbf{s}^{k_1}|\mathbf{v}\oplus\mathbf{e_{bE}}\oplus\mathbf{e_{bA}}),
\ee
where the inequality follows since $\mathbf{s}^{k_1}\to \mathbf{v}\oplus\mathbf{y} \to \mathbf{w_E}$ form a Markov chain,
and the last equality is due to the Crypto lemma \cite{erez2} and the fact that the probability distribution of
$\mathbf{x}$ is uniform over $\{0,1\}^n$ (implying that $\mathbf{x}\oplus\mathbf{e_{bE}}$ is independent
of $(\mathbf{s}^{k_1},~\mathbf{v}\oplus\mathbf{e_{bE}}\oplus\mathbf{e_{bA}})$).
Hence Eve's equivocation can only increase because of the imperfect forward channels.
\end{proof}

The impact of the forward channel finite transmission rate on the overall achievable rates is reflected in a scaling
of the \emph{unscaled} rates by the rate used over the forward link $R_{AB,fb}$.
That is, a sequence of $m_1$ bits carrying $k_1=nR_{s,u}/H_s$ secret symbols is mapped to an $n$-sequence
$\mathbf{v}$ by  Alice's secrecy encoder, such that $\frac{m_1}{n}\approx R_{t,u}$.
Next, Alice computes $\mathbf{v}\oplus\mathbf{y}$, and feeds this signal to the channel encoder.
Since $\mathbf{v}\oplus\mathbf{y}$ is a sequence of i.i.d. uniform bits (as shown in Section \ref{5sect2}),
its error free transmission requires an approximate number of $M=\frac{n}{R_{AB,fb}}$ channel uses.
Hence, the $m_1$ source bits are transmitted in $M$ channel uses.

An additional number of $n$ channel uses have to be considered for the transmission of the required feedback
sequence $\mathbf{x}$.
Noting that $\frac{n}{M+n}=\frac{R_{AB,fb}}{R_{AB,fb}+1}$, we can state the following result.

\begin{thm}\label{5theoremrateequivocregion}
For any $\nu_0$, by choosing $\nu$ such that $\nu_0>\max\{\nu, 2\nu-\nu^2\}$, we can find
a code -- comprising the original $(n,k_1,d,\overline{P_{e,1}})$ secrecy code, the
forward $(M,n,P_{e,c})$ channel code and the feedback -- which encodes the $k_1$-sequence $\mathbf{s}^{k_1}$ into the $M$-sequence $\mathbf{w_A^M}$,
such that if Bob receives $\mathbf{w_B^M}$ and Eve receives $\mathbf{w_E^M}$, we have
$\frac{k_1H_s}{M+n}\geq \frac{n}{M+n}R-\nu_0$,
$\Delta_1\geq d-\nu_0$,
and
$\overline{P_{e,1}}\leq \nu_0$,
as long as
\be
0\leq \frac{n}{M+n}R\leq \frac{R_{AB,fb}}{R_{AB,fb}+1}R_{t,u},
\ee
\be
0\leq d\leq H_s,
\ee
\be
\frac{n}{M+n}Rd\leq H_s \frac{R_{AB,fb}}{R_{AB,fb}+1}R_{s,u}.
\ee

This yields an overall secrecy rate of
\be
R_{s,0}=R_{s,u} \frac{R_{AB,fb}}{R_{AB,fb}+1}.
\ee
\end{thm}

\begin{proof}
The proof follows from Propositions \ref{prop1IIB} and \ref{prop2IIB}.
\end{proof}


\section{The First Approach: Eavesdropper's Forward Channel Less Noisy than Legitimate Receiver's Channel} \label{5sect4}

In this section we show a first approach to increasing the secrecy rate by using our feedback-based scheme.
We prove that it can achieve a strictly positive secrecy rate and a non-empty rate-equivocation region even if the
eavesdropper's forward channel $A\to E$ is less noisy than the legitimate receiver's channel $A\to B$.
The case when $A\to B$ is less noisy than $A\to E$ is studied in Section \ref{5sect5}.

If Eve's forward channel is less noisy than Bob's forward channel, or equivalently $\delta_f\leq \epsilon_f$,
then no messages can be transmitted at any level of secrecy over the $A\to B$ channel by Wyner's method \cite{wyner1}.
If we take the forwarding rate $R_{AB,fb}$ arbitrarily close to Bob's forward channel capacity $C_{AB}$, we obtain the
following result which is a straightforward adaptation of Theorem \ref{5theoremrateequivocregion}.

\begin{corol}\label{5corolary1}
For any $\nu_0>0$ there exists a code which encodes the $k$-sequence $\mathbf{s}^{k_1}$ into the $M$-sequence $\mathbf{w_A^M}$,
such that if Bob receives $\mathbf{w_B^M}$ and Eve receives $\mathbf{w_E^M}$, we have
$\frac{k_1H_s}{M+n}\geq \frac{n}{M+n}R-\nu_0$,
$\Delta_1\geq d-\nu_0$,
and
$\overline{P_e}\leq \nu_0$,
as long as
\be
0\leq \frac{n}{M+n}R\leq \frac{C_{AB}}{C_{AB}+1}R_{t,u},
\ee
\be
0\leq d\leq H_s,
\ee
\be
\frac{n}{M+n}Rd\leq H_s \frac{C_{AB}}{C_{AB}+1}R_{s,u}.
\ee

This yields an overall secrecy rate of
\be
R_{s,0}=R_{s,u} \frac{C_{AB}}{C_{AB}+1}.
\ee
\end{corol}

The following remark is in order. Maurer's ``secrecy capacity with public discussion'' \cite{maur1} is upper-bounded as follows:
\be
\widehat{C}_s(P_{YZ|X})\leq \max_{P_X}I(X;Y|Z)
\ee
where $X$, $Y$ and $Z$ denote the input and the outputs of the non-perfect channel (in our case the input to feedback channel
at Bob and the outputs at Alice and Eve, respectively), and $P_X$ denotes the probability distribution of $X$ input. It is also
noted in \cite{maur1} that in the case of binary symmetric channels, the upper-bound is achieved. For our case, this means that
the \emph{unscaled} secrecy rate $R_{s,u}=h(\epsilon_b+\delta_b-2\epsilon_b\delta_b)-h(\epsilon_b)$ can be increased no further.

However, for a practical system with imperfect forward channels, the objective should be to maximize the \emph{overall}
secrecy rate rather than the \emph{unscaled} secrecy rate. In the remainder of this section we provide a simple example to
prove that by altering the feedback sequence we can increase the overall secrecy rate of the system over the value
\be
R_{s,0}=\left[h(\epsilon_b+\delta_b-2\epsilon_b\delta_b)-h(\epsilon_b)\right] \frac{C_{AB}}{C_{AB}+1}
\ee
provided by the maximization of the unscaled secrecy rate.

\vspace*{4pt}
{\bf Processing the feedback sequence improves performance}\vspace*{4pt}

So far we assumed that the feedback i.i.d. uniform
sequence of bits $\mathbf{x}$ is transmitted by Bob with
no further processing.

\begin{figure}[]
\centering
\includegraphics[scale=0.8]{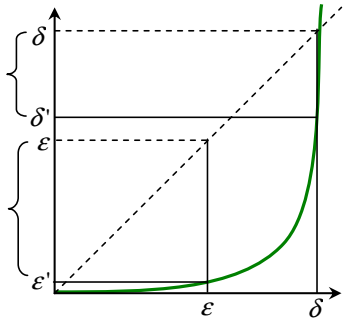}
\caption{The operator corresponding to the repetition coding preprocessing.}\label{5figoperators}
\end{figure}

Further processing of the feedback sequence results in equivalent feedback
channels with altered error probabilities.
Although the overall achievable secrecy rate depends on the rate
at which the feedback is transmitted, an error and rate
reducing encoding/decoding scheme for the feedback sequence implemented
among the three parties can improve the system's performance.  One such
simple scheme, which preserves the independence between the symbols of $\mathbf{y}$
after decoding, is obtained if Bob encodes the feedback sequence $\mathbf{x}$ using
repetition coding of rate $1/N$, and Alice and Eve employ  the optimal decoding scheme,
which is majority decoding.
The scheme results in equivalent BSCs with crossover probabilities

\begin{equation}\label{5preproc21}
\epsilon_b'= \sum_{i=K+1}^{2K+1} \left(\begin{array}{c}2K+1\\i\end{array}\right)\epsilon_b^i(1-\epsilon_b)^{2K+1-i}
\end{equation}
and
\begin{equation}\label{5preproc22}
\delta_b'= \sum_{i=k+1}^{2K+1} \left(\begin{array}{c}2K+1\\i\end{array}\right)\delta_b^i(1-\delta_b)^{2K+1-i},
\end{equation}
where $N=2K+1$ if $N$ is odd and $N=2K+2$ if $N$ is even, and $K\geq 0$.

The optimum $N$ that maximizes the overall secrecy rate can be obtained
numerically. The improvement in the overall secrecy rate due to
repetition coding, as well as the optimal choice of $N$ will be shown
in Figure \ref{5fig01} and \ref{5fig02} of Section \ref{5sect5}.
However at this point we note that a processing
method that decreases equivalent crossover probabilities is better when
$\epsilon_b$ is decreased more than $\delta_b$, i.e. when the strength
of Bob's channel is increased relative to that of Eve's. By inspecting
(\ref{5preproc21}) and (\ref{5preproc22}), we notice that the operator
corresponding to our preprocessing method is exponential. It is
therefore expected that the method gives better results when
$\epsilon_b<\delta_b$, as can be seen from Figure \ref{5figoperators}
(this phenomenon is indeed observed in our numerical results of Section \ref{5sect5}) .

Although the above result may seem counter-intuitive (in light of Maurer's
Theorem 4 \cite{maur1}), the improvement in our case results exactly from the imperfection of the
forward channels, which translates to scaling coefficients for all achievable rates, as
shown in Section \ref{5sect3}.

Note that if a  rate $1/N$ repetition coding is used for the transmission of the feedback sequence,
the total number of channel uses needed for feedback is $Nn$, leading to the overall secrecy rate
\be\label{5Ncorrection1}
R_{s,c}=\frac{nR_{s,u}}{n/R_{AB,fb}+nN}=R_{s,u} \frac{R_{AB,fb}}{NR_{AB,fb}+1}.
\ee 
The unscaled secrecy rate $R_{s,u}$ increases with $N$, while the correction factor $\frac{C_{AB}}{NC_{AB}+1}$
decreases with $N$, hence the need to find the optimal value of $N$ that maximizes $R_{s,c}$.

\textbf{Some numerical results}

Since the secrecy rate is simpler to represent than the rate-equivocation region,
throughout this paper we focus on illustrating the improvements in the achievable secrecy rate due to feedback.
We first consider a model in which the forward channels have crossover probabilities
$\epsilon_f=0.02$ and $\delta_f=0.01$, respectively.
In this scenario, Wyner's scheme cannot
deliver a secret message from Alice to Bob at any positive rate. However, the secrecy rates achievable by our feedback based scheme 
(in Figure \ref{5fig7}) are strictly positive (except in the pathological
cases when $\delta_b=0$ or $\epsilon_b=0.5$).

\begin{figure}[]
\centering
\includegraphics[scale=0.45]{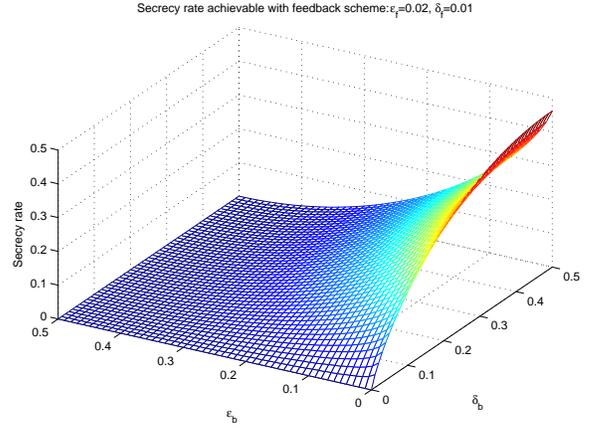}
\caption{Overall secrecy rate achievable by our feedback scheme for $\epsilon_f=0.02$ and $\delta_f=0.01$.}\label{5fig7}
\end{figure}

In Figures \ref{5fig01} and \ref{5fig02} we show the additional improvement in the overall achievable secrecy rate
obtained if we use repetition coding for the transmission of the feedback sequence, and the optimal repetition order $N$.
Although the improvement is marginal, it proves that Maurer's upper bound on the secrecy capacity with public
discussion \cite{maur1} does not hold if the forward channels are imperfect.

\begin{figure}[]
\centering
\includegraphics[scale=0.5]{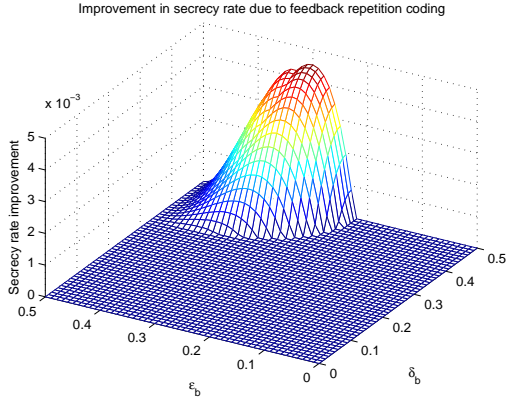}
\caption{Secrecy rate improvement due to feedback repetition coding.}\label{5fig01}
\end{figure}

\begin{figure}[]
\centering
\includegraphics[scale=0.5]{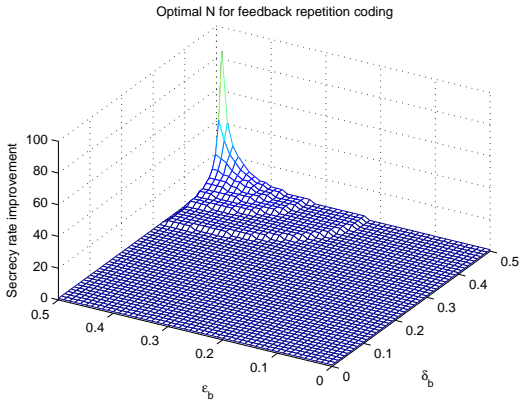}
\caption{The optimal value of $N$ for feedback repetition coding.}\label{5fig02}
\end{figure}


\section{The Second Approach: Legitimate Receiver's Forward Channel Less Noisy than Eavesdropper's Channel} \label{5sect5}

If $\epsilon_f< \delta_f$, a non-empty rate-equivocation region and a strictly positive secrecy rate less than
$C_s=C_{AB}-C_{AE}$ are asymptotically achievable without feedback \cite{wyner1}.
In this section we show that even under this scenario, sometimes feedback can improve the achievable secrecy rate.
For example, when $C_s$ is small compared to $C_{AB}$,and the unscaled secrecy rate achievable with feedback $R_{s,u}$
is relatively large (i.e. when the channel $B\to A$ is significantly better than the channel $B\to E$, while
the channel $A\to B$ is only slightly better than the channel $A\to E$) , we can have $C_s<R_{s,u} \frac{C_{AB}}{C_{AB}+1}$.

However, in general, neither Wyner's original scheme, nor our feedback based scheme is optimal. Instead, as we shall see shortly,
encoding a combination of a secret message and a feedback-processed message into the forwarded sequence $\mathbf{w_A}$
can achieve a higher overall secrecy rate.

The method behind the direct part of Wyner's Theorem 2 \cite{wyner1} assumes the transmission of $m_2$ bits, containing
$k_2=n_2C_s$ secret bits, by mapping the $k_2$-bit secret message $\mathbf{s}^{k_2}$ to a specific coset. The rest of
$m_2-k_2$ bits correspond to the index of the randomly picked coset representative which is transmitted.
Since Bob can decode the transmitted codeword perfectly, he has access to all $m_2$ bits.
The $m_2-k_2$ non-secret bits are neither secret to, nor can they be decoded by Eve without errors \cite{csisz1}.
It was assumed in \cite{wyner1} that these bits are picked randomly (according to
a uniform distribution) and carry no information. In their extension of Wyner's work, Csiszar and Korner
\cite{csisz1} observe that these bits can actually be picked according to the output message of a uniform
source of entropy $H_x=m_2-k_2$, which can carry useful information for Bob \cite{csisz1}.

At a first glance, it would appear that by encoding the message $\mathbf{v}\oplus\mathbf{y}$ into the $m_2-k_2$
non-secret bits, we could transmit it asymptotically error free to Bob, at a rate arbitrarily close to
$C_{AB}-C_s=C_{AE}$, in addition to the original secret message $\mathbf{s}^{k_2}$. In this case, even if Eve had perfect access
to these bits (which she has not), the equivocation of both secret messages would be preserved.
This argument is the starting point of the proposed mixed secrecy scheme of \cite{gunduz} (see Appendix \ref{5appgunduzwrong}
for more remarks on \cite{gunduz}). Unfortunately, the
argument above is false. By using the sequence $\mathbf{v}\oplus\mathbf{y}=\mathbf{v}\oplus\mathbf{x}\oplus\mathbf{e_{bA}}$
to pick the coset representative to be transmitted over the forward channel, the equivocation of the secret message $\mathbf{s}^{k_2}$
encoded into the other $k_2$ bits is compromised. 
As shown in Appendix \ref{5appgunduzwrong}, this happens because Eve has access to a distorted version of the feedback sequence
$\mathbf{x}\oplus \mathbf{e_{bE}}$, which is correlated with $\mathbf{v}\oplus\mathbf{y}$. 

Therefore we need an encoding technique in which Eve's information about the message $\mathbf{v}\oplus\mathbf{y}$,
obtained through $\mathbf{x}\oplus \mathbf{e_{bE}}$, does not influence the secrecy of $\mathbf{s}^{k_2}$. Such
a technique is readily provided by \cite{csisz1}. The encoding technique of \cite{csisz1} aims at transmitting not
only a secret message from Alice to Bob, but also a common message from Alice to both Bob and Eve. The code is
designed following a 2-cycle maximal construction idea. First, a sub-code which can carry information reliably over
both channels, at a sub-optimal rate is picked for the common message. Other codewords are then added to the sub-code
(in two cycles) -- such that Bob can distinguish between any two codewords, while Eve can only distinguish between any
two codewords corresponding to the same secret message -- until no more such codewords exist.

Adapting this strategy to our case, we can treat the sequence $\mathbf{v}\oplus\mathbf{y}$ as a common
message, intended for both Bob and Eve. In addition to the common message, we can also transmit a secret message $\mathbf{s}^{k_2}$ 
to Bob. Since the common message is designed to be perfectly decoded by Eve, the additional information
contained in $\mathbf{x}\oplus \mathbf{e_{bE}}$ cannot compromise the secrecy of $\mathbf{s}^{k_2}$. The drawback is that
the transmission of a common message decreases the rate at which the secret message $\mathbf{s}^{k_2}$ can be conveyed
to Bob \cite{csisz1}. However, the transmission of an additional secret message $\mathbf{s}^{k_1}$, encoded in the sequence
$\mathbf{v}$, can make up for this loss and, in many circumstances, bring noticeable improvements over Wyner's scheme \cite{wyner1}.

In order to pursue this path, we first need to establish what is the optimal tradeoff between the common message rate
and the secret message rate. Denote by $W_A$, $W_B$ and $W_E$ the input to
the forward channel and the outputs at Bob and Eve, respectively.
According to Theorem 1 of \cite{csisz1}, the two rates have to satisfy:
\be\label{5relRe0}
R_e\leq I(V;W_B|U)-I(V;W_E|U),
\ee
\be\label{5relRc0}
R_c\leq \min[I(U;W_B),I(U;W_E)],
\ee
where $R_e$ is the secret message rate, $R_c$ is the common message rate, and $U$ and $V$ are two auxiliary random variables
such that $U\to V\to W_A\to W_B,W_E$ form a Markov chain.

For our special BSC case, and under the scenario where $\epsilon_f< \delta_f$, we can further simplify (\ref{5relRc0}):
\be\label{5relRc}
R_c\leq I(U;W_E).
\ee
Following the proof of Corollary 3 in \cite{csisz1}, we can write (\ref{5relRe0}) as:
\be\label{5relRe}
R_e\leq I(V;W_B|U)-I(V;W_E|U)=\nonumber\\
=I(V;W_B)-I(V;W_E)-[I(U;W_B)-I(U;W_E)]=\nonumber\\
=[I(W_A;W_B)-I(W_A;W_E)]-\nonumber\\
-[I(W_A;W_B|V)-I(W_A;W_E|V)]-\nonumber\\
-[I(U;W_B)-I(U;W_E)],
\ee
where the equalities follow from the fact that if $X\to Y\to Z$ form a Markov chain, then $I(Y;Z)=I(X;Z)+I(Y;Z|X)$
(Lemma 1 in \cite{csisz1}). Note that the term $[I(W_A;W_B|V)-I(W_A;W_E|V)]$ is always positive \cite{csisz1},
and is minimized for $V=W_A$. The condition in (\ref{5relRe}) is thus reduced to
\be
R_e\leq [I(W_A;W_B)-I(W_A;W_E)]-\nonumber\\-[I(U;W_B)-I(U;W_E)],
\ee
or equivalently
\be
R_e\leq I(W_A;W_B|U)-I(W_A;W_E|U).
\ee

At this point we are looking for the auxiliary random variable $U$, and its relationship with the channel input
random variable $W_A$, that achieve the points on the boundary of the $(R_e,R_c)$ region described above.
The only information about $U$ that
is provided in \cite{csisz1}, is that its alphabet size may, without loss of generality, be assumed to be at most
three letters larger than the alphabet of
$W_A$ (in our binary case, the alphabet of $U$ would have at most five letters).

The following three results(two propositions and one conjecture) greatly simplify the search for the optimal auxiliary
random variable and channel. The two propositions are proved, and the arguments behind the conjecture are
presented, in Appendix \ref{5apptheorembinaryu}.

\begin{prop}\label{prop3letters}
The optimal auxiliary random variable $U$ can be defined, without loss of optimality,
over a three-dimensional alphabet.
\end{prop}

\begin{conj}\label{conj2letters}
The optimal auxiliary random variable $U$ can be defined, without loss of optimality,
over a binary alphabet.
\end{conj}

\begin{prop}\label{propbsc}
If $U$ is considered binary, then its optimal distribution over its two-dimensional alphabet (pick it as $\{0,1\}$
for convenience) is a uniform one. Moreover, the optimal auxiliary channel that links $U$ to the physical channel
input $W_A$ is a simple binary symmetric channel.
\end{prop}

In the remainder of this paper we shall assume that $U$ is a binary, uniform random variable, linked to $W_A$
through a BSC of crossover probability $\gamma$. Note that even if Conjecture \ref{conj2letters} were false,
this assumption would not interfere with the achievability of our secrecy rates. Instead, our rates would
lower-bound the secrecy rates achievable under the optimal distribution of a ternary $U$ and the corresponding
optimal auxiliary channel between $U$ and $W_A$.

Once we pick the auxiliary channel crossover probability $\gamma$
we can compute
\be\label{5rel2rc}
R_c^*=1-h(\gamma+\delta_f-2\gamma\delta_f)
\ee
and
\be\label{5rel2re}
R_e^*=[h(\delta_f)-h(\epsilon_f)]-\nonumber\\-[h(\gamma+\delta_f-2\gamma\delta_f)-h(\gamma+\epsilon_f-2\gamma\epsilon_f)].
\ee

Similar arguments to those in the previous section apply to show that the messages $\mathbf{v}\oplus\mathbf{y}$,
containing the secret message $\mathbf{s}^{k_1}$, can now be transmitted to Bob asymptotically error free at
a rate arbitrarily close to $R_c^*$, in the form of a common message. In addition, another secret message $\mathbf{s}^{k_2}$
can be transmitted simultaneously to Bob at rate close to $R_e^*$. In the remainder of this section we calculate the resulting
overall secrecy rate.

\begin{prop}
If the legitimate receiver's channel is less noisy than the eavesdropper's channel,
the secrecy rate
\be\label{5relrsofinal}
R_{s,0}=\max\left[\max_{\gamma} \frac{(R_e^*+R_c^*R_{s,u})}{R_c^*+1},\frac{C_{AB}R_{s,u}}{C_{AB}+1}\right]
\ee
is achievable by our feedback-based scheme, where $R_e^*$ and $R_c^*$ are given by (\ref{5rel2rc}) and
(\ref{5rel2re}), respectively.
\end{prop}

\begin{proof}
Define the equivocations $\Delta_1=\frac{1}{k_1}H(\mathbf{s}^{k_1}|\mathbf{w_E^M}, \mathbf{x}^n+\mathbf{e_{bE}}^n)$ and
$\Delta_2=\frac{1}{k_2}H(\mathbf{s}^{k_2}|\mathbf{w_E^M})$,
where $\mathbf{s}^{k_2}$ is the $k_2$-sequence of secret source symbols that are encoded in the codeword
$\mathbf{w_A^M}$ as a secret message, and $\mathbf{s}^{k_1}$ is a distinct $k_1$-sequence of secret source
symbols that are encoded in the sequence $\mathbf{v}\oplus\mathbf{y}$ by our feedback scheme.
The sequence $\mathbf{v}\oplus\mathbf{y}$ is in turn mapped into the same codeword $\mathbf{w_A^M}$ as a common message.
The transmitted codeword $\mathbf{w_A^M}$ is received by Eve as $\mathbf{w_E^M}$.
We know that for any $\nu>0$ there exists such an encoding technique which satisfies
\be\label{5brrrr1}
\begin{array}{lcr}
\frac{k_2 H_s}{M}\geq R_2-\nu,
&
\Delta_2\geq d_2-\nu,
&
\overline{P_{e,2}}\leq \nu,
\end{array}
\ee
as long as
\be
\begin{array}{lcr}
0\leq R_2\leq C_{AB},
&
0\leq d_2\leq H_s,
&
R_2d_2\leq H_s R_e^*,
\end{array}
\ee
and
\be\label{5brrrr2}
\begin{array}{lcr}
\frac{k_1 H_s}{M+n}\geq R_1-\nu,
&
\Delta_1\geq d_1-\nu,
&
\overline{P_{e,1}}\leq \nu,
\end{array}
\ee
as long as
\be
\begin{array}{lr}
0\leq R_1\leq R_{t,u}\frac{R_c^*}{R_c^*+1},
&
0\leq d_1\leq H_s,\\~&
R_1d_1\leq H_s R_{s,u}\frac{R_c^*}{R_c^*+1}.
\end{array}
\ee

The equivocation of the secret message at Eve is now defined as:
\be
\Delta=\frac{1}{k_1+k_2}H(\mathbf{s}^{k_1},\mathbf{s}^{k_2}|\mathbf{w_E^M}, \mathbf{x}^n+\mathbf{e_{be}}^n).
\ee
Since $\mathbf{s}^{k_1}$ and $\mathbf{s}^{k_2}$ are independent, we can write
\be
\Delta=\frac{k_1}{k_1+k_2}\Delta_1+\frac{k_2}{k_1+k_2}\Delta_2.
\ee
Note that the overall rate at which the secret source is transmitted is now $\frac{(k_1+k_2)H_s}{M+n}$.
Therefore, a correction of $\frac{M}{M+n}$ has to be applied to the rate $R_2$.
As a result, the rate-equivocation pair $(R,d)$ is achievable if $R=\min\{\frac{M}{M+n}R_2+R_1,C_{AB}\}$
and $d=\frac{k_1}{k_1+k_2}d_1+\frac{k_2}{k_1+k_2}d_2$.
Note that this implies $R<C_{AB}$ and $d<H_s$.

Also, due to Corollary \ref{5cor1wyner}, if $M$ is large enough, we have $k_2 d_2 \to MR_e^*$ and $k_1 d_1 \to (M+n)\frac{R_c^*}{R_c^*+1}R_{s,u}$.
Due to (\ref{5brrrr1}) and (\ref{5brrrr2}) we also have
$k_1+k_2 \to \frac{M+n}{H_s}(\frac{M}{M+n}R_2+R_1)$, so we can write
\be
d=\frac{k_1d_1+k_2d_2}{k_1+k_2}
\to H_s\frac{1}{R}\left[\frac{M}{M+n}R_e^*+\frac{R_c^*R_{s,u}}{R_c^*+1}\right].
\ee
Recall that for this case, the $n$-sequence $\mathbf{v}\oplus \mathbf{y}$ is encoded in the $M$-sequence $\mathbf{w_A^M}$
at rate $R_c^*$. Thus, $\frac{n}{M}=R_c^*$, which implies $\frac{M}{M+n}=\frac{1}{R_c^*+1}$, leading to
\be
d \to H_s\frac{1}{R}\frac{R_e^*+R_c^*R_{s,u}}{R_c^*+1}
\ee 
and
\be\label{5kjhfgjvuy}
Rd\to H_s\frac{(R_e^*+R_c^*R_{s,u})}{R_c^*+1}.
\ee
Note that the condition for achieving equality asymptotically (as $M\to \infty$) in (\ref{5kjhfgjvuy}) above is that
the two levels of secrecy operate at $R_2d_2 = H_s C_e^*$ and $R_1d_1= H_s \frac{R_c^*}{R_c^*+1}R_{s,u}$ respectively.

Several comments are in order.
If $\gamma=0$, we obtain $R_c^*=C_{AE}$, and $R_e^*=0$. However in this case, since no secret message is transmitted
directly by Wyner's scheme, we can safely transmit the feedback-processed message at a rate $R_{AB,fb}=C_{AB}$
just like in (Subsection \ref{5sect4}). This discontinuity in $\gamma=0$ is why in (\ref{5relrsofinal}) we have to
compare the result of the maximization over $\gamma$ (corresponding to the mixed scheme) with the rate achieved by
the pure feedback scheme. If $\gamma=0.5$, we have $R_c^*=0$, and $R_e^*=C_{AB}-C_{AE}=C_s$, resulting in Wyner's
original scheme \cite{wyner1} -- hence no discontinuity in $\gamma=0.5$. Any value of $\gamma$ in the open
interval $(0,0.5)$ results in a combination of the two schemes.
\end{proof}

\textbf{Some more numerical results}

To illustrate the performance of our second approach to implementing the feedback-based secrecy scheme,
we consider a model in which the forward channels have crossover probabilities
$\epsilon_f=0.01$ and $\delta_f=0.02$, respectively.
The secrecy rate achievable by Wyner's original scheme is $C_s=0.06$.

In Figure \ref{5fig2} we show the overall achievable secrecy rate when we use our
feedback scheme, for different values of the crossover probabilities characterizing the feedback channels.
The corresponding optimal value of the parameter $\gamma$ is given in Figure \ref{5fig3}. Recall that
whenever $\gamma=0.5$, our feedback scheme reduces to Wyner's scheme, and hence the achievable secrecy rate
is $C_s$. Also, when $\gamma=0$, our scheme uses the whole capacity $C_{AB}$ of Bob's forward channel to
convey a secret message encoded with the help of the feedback sequence (no additional directly encoded
secret message is present). The improvements are significant.

\begin{figure}[]
\centering
\includegraphics[scale=0.45]{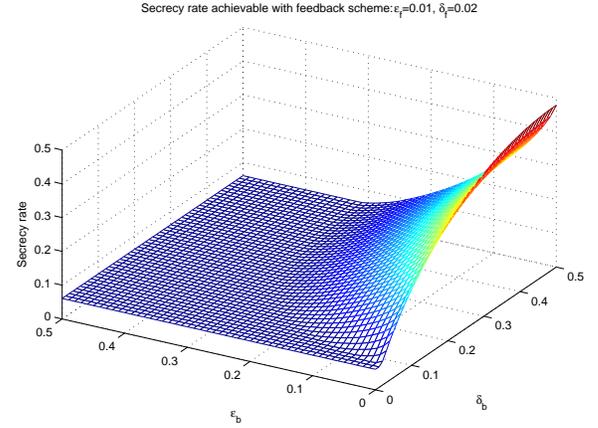}
\caption{Secrecy rate achievable by the feedback scheme for $\epsilon_f=0.01$ and $\delta_f=0.02$.}\label{5fig2}
\end{figure}

\begin{figure}[]
\centering
\includegraphics[scale=0.45]{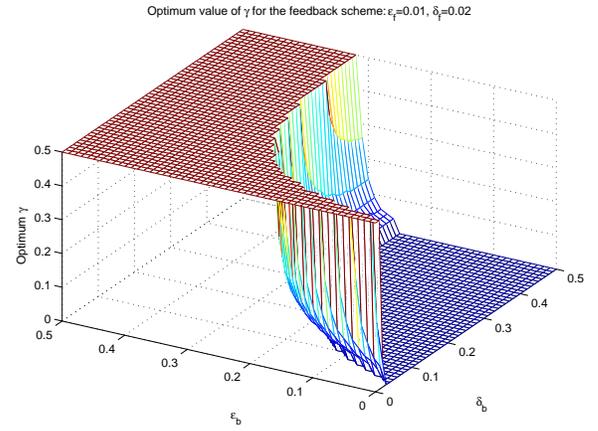}
\caption{The optimal value of $\gamma$ for the feedback scheme when $\epsilon_f=0.01$ and $\delta_f=0.02$.}\label{5fig3}
\end{figure}


\section{The Third Approach: The Reversed Feedback Scheme}\label{5sect6}

The feedback-based scheme discussed in the previous section encodes two secret messages
into the sequence transmitted over the forward channel.
The main idea behind this construction is based on the capability of the legitimate
transmitter (Alice) to transmit two types of messages simultaneously \cite{csisz1}: a first secret message to Bob, and
a common message to both Bob and Eve. In our case, the common message carries a second secret message,
the encoding of which is based on artificially degrading Eve's equivalent channel by the use of a feedback sequence.
But on a deeper level, the encoding of both secret messages uses the same principle developed in \cite{wyner1},
and none of them uses an explicit secret key. 

In this section, we discuss another approach to increasing the secrecy rate, namely when the feedback-based scheme is used on the reversed
channel (in the sense that the secret message encoded with the help of our feedback-based scheme is now transmitted from Bob to Alice
instead of from Alice to Bob) to send a secret
key from Bob to Alice, much like in \cite{ardetsan} and \cite{gunduz} (in fact the scenarios of 
\cite{ardetsan} and the correct part of \cite{gunduz} can be considered as special cases of our reversed mixed feedback
scheme.). Alice can subsequently use this secret key as a one-time pad, for transmitting a secret
message of the same entropy \cite{schneier} to Bob.

Although this new protocol requires more bandwidth than the previous one, it can sometimes
achieve better overall performance in terms of rate-equivocation region and secrecy rate.
However, this can only happen under the (necessary but not sufficient) condition that
the rate at which the secret key is transmitted from Bob to Alice exceeds the secrecy rates
achievable by the original feedback scheme.

Denote by $R_{s,p}$ the supremum of the rates at which Bob can transmit a secret key (or a one-time pad) to Alice
by using the feedback scheme developed in the previous section on the reversed channel.
Note that $R_{s,p}$ can be obtained from the expression of $R_{s,0}$ in (\ref{5relrsofinal}) by replacing
$\epsilon_f$ by $\epsilon_b$, $\delta_f$ by $\delta_b$, and vice versa.

To acquire this secret key, Alice and Bob engage in a protocol which is the reversed version of the one described in
the previous sections. Alice broadcasts a random feedback sequence of $n$ bits. Bob can then encode $k_1$ secret bits
into an $n$-sequence, which is added $\mod 2$ to Bob's received feedback sequence, and then the result is further
encoded into an $M$-sequence for asymptotically error free transmission over the $B\to A$ and $B\to E$ channels.

If $C_{BA}>C_{BE}$, the same $M$-sequence can carry an additional secret message of $k_2$ bits.
A number of $M+n$ channel uses are thus required for the transmission of a $k_r=k_1+k_2$-bit secret key
$\mathbf{r}^{k_r}$, and are accounted for in the expression of $R_{s,p}$ (that is, $R_{s,p}=\frac{k_r}{M+n}$).

After adding the secret key $\mathbf{r}^{k_r}$ to a secret message $\mathbf{s_r}^{k_r}$ of her own (also a ${k_r}$-bit sequence),
Alice encodes the result into an $M'$-sequence for the forward channel. Note here that because Alice uses a secret key,
the secrecy of $\mathbf{s_r}^{k_r}$ is preserved (by the Crypto lemma \cite{erez2}) even if Eve has perfect access to
the resulting ${k_r}$-bit sum sequence $\mathbf{r}^{k_r}\oplus \mathbf{s_r}^{k_r}$.

At this point, Alice could choose to encrypt everything she transmits to Bob. However, that strategy would require
the generation of a long secret key, and hence cause a large rate loss due to feedback -- recall that in our results we
count the bandwidth expenditure due to feedback. Instead, a mixed secrecy encoding strategy on the forward link may be optimal.
For example, a special adaptation of our reversed feedback scheme is possible when $C_{AB}>C_{AE}$. Recall that in Section \ref{5sect5} we
made a comment about the possibility to transmit a secret message, encoded in the cosets of a code,
at a rate arbitrarily close to the secrecy capacity $C_s=\max \{C_{AB}-C_{AE},0\}$, while using the feedback-processed
sequence $\mathbf{v}\oplus \mathbf{y}$ (that was carrying a separate secret message) for selecting the exact coset
representative to be transmitted. In Section \ref{5sect5} this was not possible due to the fact that Eve had some
information about $\mathbf{v}\oplus \mathbf{y}$, from its received feedback sequence $\mathbf{x}\oplus \mathbf{e_{bE}}$.
In the present scenario, however, the message $\mathbf{r}^{k_r}\oplus \mathbf{s}^{k_r}$ is totally unknown to Eve, and can be safely
used for selecting the coset representative.

Thus, a first $k_0$-bit secret message -- denote it by $\mathbf{s_0}^{k_0}$ -- can be transmitted from Alice
to Bob using Wyner's original scheme \cite{wyner1}, at a rate $\frac{k_0}{M'}\simeq C_s$. A second secret message
$\mathbf{s_r}^{k_r}$ can be transmitted at a rate $\frac{k_r}{M'}\simeq C_F$ (we denoted $C_F=\min\{C_{AB},C_{AE}\}$)
by using the secret key $\mathbf{r}^{k_r}$ generated through a reversed feedback scheme.

With this notation, and taking into account all $n+M+M'$ channel uses involved in the protocol
(i.e. $n$ for the reversed feedback sequence from Alice to Bob, $M$ for the transmission of the secret key from Bob
to Alice, and $M'$ for the transmission of the secret message from Alice to Bob), we can write the overall achievable secrecy rate as
\be\label{5revfbsecrecyrate}
R_{s,rf}=\frac{k_0+k_r}{n+M+M'}=\nonumber\\
=\frac{M'}{n+M+M'}(C_s+C_F)=C_{AB}\frac{R_{s,p}}{C_F+R_{s,p}},
\ee
where in the second equality we used the fact that $C_F+C_s=C_{AB}$ and that
\be
\frac{M'}{n+M+M'}=\frac{k_r/(n+M)}{k_r/M'+k_r/(n+M)}=\frac{R_{s,p}}{C_F+R_{s,p}}.
\ee

An observation is now in order. Although a secret key of length equal to that of the transmitted message
may be generated by our reversed feedback scheme, employing Wyner's original scheme, when possible, in
addition to the encryption by the secret key is always optimal. Indeed, Wyner's
scheme guarantees the transmission of a secret message without wasting any resources other than the $M'$ bits
of the forward channel sequence, while encrypting a message by a secret key generated as above requires additional
resources that grow linearly with the size of the secret key. For example, generating a secret key
long enough to encrypt the whole secret message (of size $M' C_{AB}$ bits) yields an achievable secrecy rate
equal to $C_{AB}\frac{R_{s,p}}{C_{AB}+R_{s,p}}$, which is always less than the secrecy rate $C_{AB}\frac{R_{s,p}}{C_F+R_{s,p}}$
above, achieved by the mixed scheme.

\textbf{Numerical Results}

For the first data set, of Section \ref{5sect4}, ($\epsilon_f=0.02$ and $\delta_f=0.01$), the achievable secrecy rate and the
optimal $\gamma$ for the reversed feedback scheme are given in Figures \ref{5fig8} and \ref{5fig5}. 
\begin{figure}[]
\centering
\includegraphics[scale=0.45]{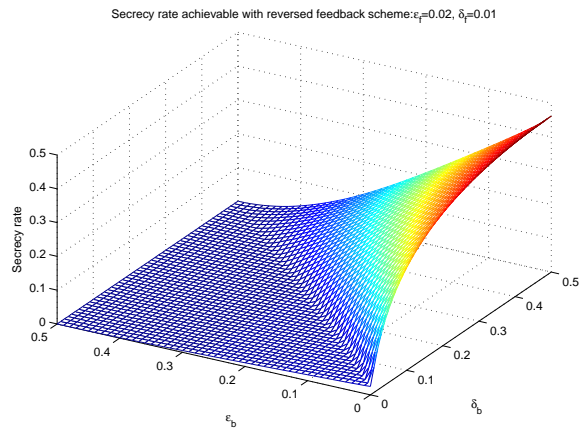}
\caption{Overall secrecy rate achievable by the reversed feedback scheme for $\epsilon_f=0.02$ and $\delta_f=0.01$.}\label{5fig8}
\end{figure}

\begin{figure}[]
\centering
\includegraphics[scale=0.45]{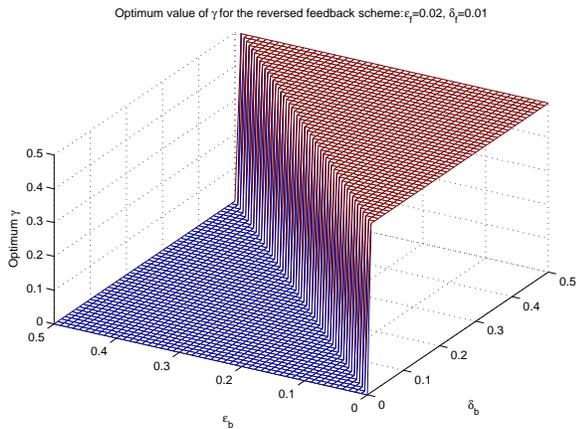}
\caption{The optimal value of $\gamma$ for the reversed feedback scheme when $\epsilon_f=0.02$ and $\delta_f=0.01$.}\label{5fig5}
\end{figure}

The improvement in the overall secrecy rate when using the reversed feedback scheme instead of
the regular feedback scheme, i.e. the function $\max\{0,R_{s,rf}-R_0\}$, is shown in Figure \ref{5fig6}.
Note that the reversed mixed feedback scheme is usually a better choice
when Eve's feedback channel is worse than Alice's (i.e. $\delta_b>\epsilon_b$).

\begin{figure}[]
\centering
\includegraphics[scale=0.45]{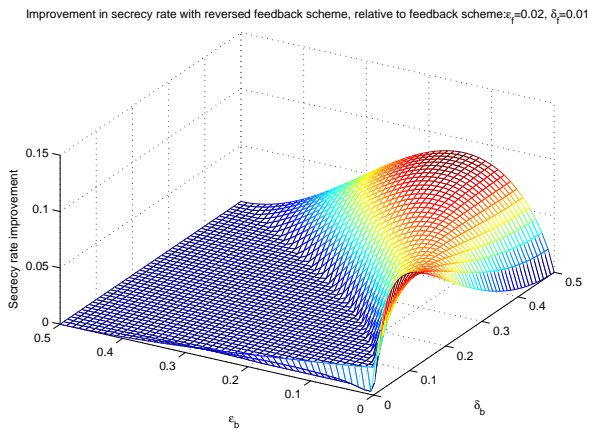}
\caption{Improvement in overall secrecy rate when using the reversed feedback scheme instead of
the regular feedback scheme: $\epsilon_f=0.02$ and $\delta_f=0.01$. Represented is the function
$\max\{0,R_{s,rf}-R_{s,0}\}$.}\label{5fig9}
\end{figure}

For the second data set, of Section \ref{5sect5}, ($\epsilon_f=0.01$ and $\delta_f=0.02$), the secrecy rate $R_{s,f}$ achievable by the
reversed mixed feedback scheme is given in Figure \ref{5fig4}, and the improvement over the regular mixed feedback scheme is
depicted in Figure \ref{5fig5}. 
\begin{figure}[]
\centering
\includegraphics[scale=0.45]{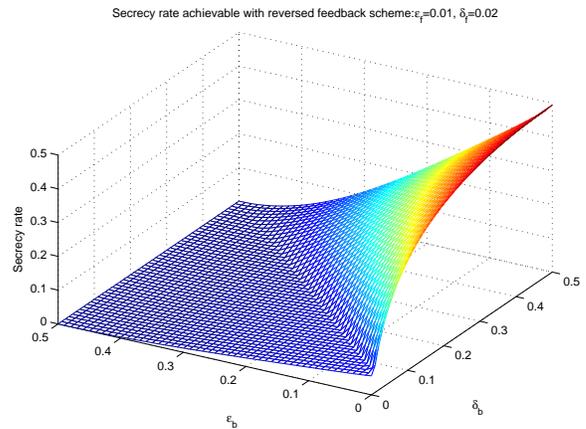}
\caption{Secrecy rate achievable by the reversed mixed feedback scheme for $\epsilon_f=0.01$ and $\delta_f=0.02$.}\label{5fig4}
\end{figure}
\begin{figure}[]
\centering
\includegraphics[scale=0.45]{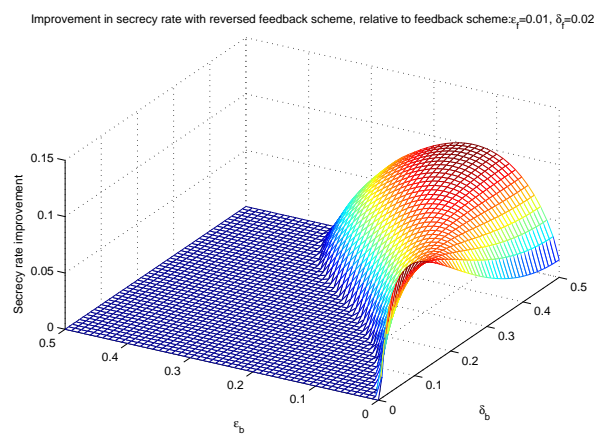}
\caption{Improvement in secrecy rate when using the reversed mixed feedback scheme instead of
the regular mixed feedback scheme: $\epsilon_f=0.01$ and $\delta_f=0.02$. Represented is the function
$\max\{0,R_{s,rf}-R_{s,0}\}$.}\label{5fig6}
\end{figure}
Once again, the reversed feedback scheme performs better when $\delta_b>\epsilon_b$.
It is also interesting to note the existence of a region in the $(\epsilon_b,\delta_b)$
plane (around the diagonal $\epsilon_b=\delta_b$), where our regular mixed feedback scheme
beats Wyner's scheme even when $\epsilon_b>\delta_b$ (see Figure \ref{5fig2}), and it also beats the reversed mixed feedback
scheme even when $\epsilon_b<\delta_b$ (see Figure \ref{5fig6}).


\section{Conclusions}\label{5sect8}

We presented a scheme that achieves a strictly positive secrecy rate
even if the eavesdropper's channel is better than the legitimate receiver's
channel, and improves the achievable secrecy rate if the eavesdropper's channel is worse.

We proposed several collaborative secrecy encoding methods, all of which use our feedback scheme.
Depending on the channel conditions,
the possible ways in which the feedback-based scheme can be used are summarized in
Table \ref{tblimplementmethods}. The term \emph{pure feedback scheme} in Table
\ref{tblimplementmethods} denotes the feedback scheme as implemented in Section \ref{5sect4}, i.e.
without being mixed with Wyner's scheme, while \emph{mixed feedback scheme} refers to the implementation
of Section \ref{5sect5}, under the optimal mixture between the pure feedback scheme and Wyner's scheme.
Similar considerations hold for the \emph{reversed pure/mixed feedback scheme} (see Section \ref{5sect6}).

\begin{table}\caption{Possible implementation of our feedback-based secrecy scheme.}\label{tblimplementmethods}
\centering
\begin{tabular}{|l|c|}
\hline
Channel conditions & Possible implementation\\
\hline\hline
$C_{BA}\leq C_{BE}$ & Pure feedback scheme\\
$C_{AB}\leq C_{AE}$ & \\
\hline
$C_{BA}>C_{BE}$ & Pure feedback scheme OR\\
$C_{AB}\leq C_{AE}$ & Reversed mixed feedback scheme\\
\hline
$C_{BA}\leq C_{BE}$ & Mixed feedback scheme OR\\
$C_{AB}>C_{AE}$ & Reversed pure feedback scheme\\
\hline
$C_{BA}>C_{BE}$ & Mixed feedback scheme OR\\
$C_{AB}>C_{AE}$ & Reversed mixed feedback scheme\\
\hline
\end{tabular}
\end{table}

Our scheme requires a new random sequence to be fed back from Bob,
for each codeword that Alice wants to send over the forward channel, in a manner similar
to the one-time pad.
We have shown that Theorem 4 in \cite{maur1}, which provides an upper bound on the achievable secrecy rate
when the public channel is error free, does not hold if this condition is not satisfied. The derivation of
such an upper bound for the more realistic scenario with imperfect public channels is still under our investigation.

The main advantage of our scheme is that it makes physical layer security protocols implementable
with only minor restrictions imposed on the eavesdropper's channel, restrictions which can be easily
ensured by perimeter defense (transmission power is low enough to guarantee a minimum error probability for
any terminal situated outside a safe perimeter).


\appendices

\section{Why the Approach of \cite{gunduz} Is Wrong}\label{5appgunduzwrong}

Since the ideas of \cite{gunduz} are closely related to our feedback secrecy encoding scheme,
and since \cite{gunduz} suffers from several subtle flaws, we dedicate this appendix
to pointing out three of these. 

First, all the rates of \cite{gunduz} are expressed without considering the expense of channel
uses due to feedback. While this may seem like a minor inconvenience as far as the forward channel rates
are involved, it becomes a problem when the forward channel rates are mixed with orthogonal feedback channel
rates, as in sections 3 and 4 of \cite{gunduz}. More specifically, the secrecy rate achievable by
Wyner's scheme on the forward channel cannot be added to the rate at which the secret key is generated
over the orthogonal feedback channel unless both channels use the exact same codeword length.

Second, even if  both the feedback and forward channels used the same
codeword length, the time sharing idea of \cite{gunduz} is questionable. It is claimed in \cite{gunduz}
that time sharing is performed between two modes of operation on the feedback channel: Wyner's regular
scheme, and our feedback secrecy scheme. With the notation of \cite{gunduz}, the two modes of operation
would normally yield secrecy rates $C_s^b=[h(\delta_b)-h(\epsilon_b)]^+$ (Wyner's scheme) and
$R_{fbs}=h(\epsilon_b+\delta_b-2\epsilon_b\delta_b)-h(\epsilon_b)$ (our feedback scheme). Thus,
the optimal time sharing between these schemes is to always use our feedback secrecy scheme
(i.e. $\alpha=0$ always in \cite{gunduz}) since $R_{fbs}>C_s^b$ regardless of the channel parameters.

Third, our secrecy feedback scheme cannot be mixed with Wyner's secrecy scheme the way that was
claimed in section 4 of \cite{gunduz}. The mixed strategy of \cite{gunduz} was inspired by some
of our results in \cite{myself6}, which are incorrect, and for which we assume full responsability.

Whenever this type of mixing is desired, special care should be taken to ensure that
Eve's information about the random feedback sequence, obtained on the feedback channel, does not compromise
the secrecy of Wyner's scheme. We have already mentioned this in Section \ref{5sect5}. In the following,
we give a more detailed explanation of this account. With the notation on Section \ref{5sect5}, consider
the secret message encoded by Wyner's scheme $\mathbf{s}^{k_2}$, the auxiliary message -- the one used
for picking the exact bin representative \cite{wyner1}, and which contains another
secret message, encoded with the use of the feedback scheme -- $\mathbf{v}\oplus \mathbf{y}$, Alice's transmitted
sequence $\mathbf{w_A^M}$ (which is a deterministic function of $\mathbf{s}^{k_2}$ and $\mathbf{v}\oplus \mathbf{y}$),
and Eve's received sequence $\mathbf{w_E^M}$.

The key to Wyner's secrecy scheme is to employ an
encoding scheme that guarantees that $H(\mathbf{v}\oplus \mathbf{y})$ is arbitrarily close to, but less than,
$I(\mathbf{w_A^M};\mathbf{w_E^M})$ \cite{wyner1}. Indeed, this is how the encoding in \cite{gunduz} is performed.

However, recall that due to the feedback scheme, Eve also has access to a distorted version $\mathbf{z}$
of Bob's feedback sequence $\mathbf{x}$. In order for Wyner's scheme to still work, we would need to have
$H(\mathbf{v}\oplus \mathbf{y}|\mathbf{z})$ arbitrarily close to $I(\mathbf{w_A^M};\mathbf{w_E^M},\mathbf{z})$.
But we can write
\be
H(\mathbf{v}\oplus \mathbf{y}|\mathbf{z})<H(\mathbf{v}\oplus \mathbf{y})\simeq\nonumber\\
\simeq I(\mathbf{w_A^M};\mathbf{w_E^M})\leq I(\mathbf{w_A^M};\mathbf{w_E^M},\mathbf{z}),
\ee
where the first inequality holds in a strict sense and follows from
the fact that, because $\mathbf{v}$ is the output of a Wyner-type channel encoder (for the artifficially
created equivalent channels -- see Section \ref{5sect2}), it cannot be uniformly distributed over
$\{0,1\}^n$, and thus $\mathbf{v}\oplus \mathbf{y}$ is not independent of $\mathbf{z}$.
The second inequality follows from the fact that $\mathbf{z}\to\mathbf{x}
\to\mathbf{y}\to\mathbf{w_A^M}\to\mathbf{w_E^M}$ form a Markov chain.


\section{The Optimal Tradeoff between the Secret Rate and the Common Rate}\label{5apptheorembinaryu}

In Section \ref{5sect5} we have already shown that for an eavesdropper channel with input (at Alice) $X$ and outputs $Y$ at
the legitimate receiver (Bob) and $Z$ at the eavesdropper (Eve), for which the Bob's channel is less noisy than Eve's
channel, a pair of one secret and one common messages can
be transmitted with asymptotically zero average error probability if and only if the rate $R_e$ of the secret message
and the rate $R_c$ of the common message satisfy
\be\label{5rel11111}
R_e\leq I(X;Y|U)-I(X;Z|U)
\ee
and
\be\label{5rel22222}
R_c\leq I(U;Z),
\ee
where $U$ is an auxiliary random variable such that $U\to X\to YZ$ form a Markov chain. This result is a straightforward
particularization of Theorem 1 in \cite{csisz1}, for the case when Bob's channel is less noisy and we are only concerned
with common and secret messages.
From an application point of view, an efficient communications system that uses the framework in \cite{csisz1} to transmit
two such messages should operate on the boundary of the $(R_e,R_c)$ rate region. For example, once $R_c$ is set
to a fixed value $R_c^*$, the system should aim to use the maximum secrecy rate $R_e$ available under these circumstances. This is
equivalent to finding the optimal auxiliary random variable $U$, and the optimal relation (we shall henceforth denote this
relation by the term ``channel'') between $U$ and $X$, that maximize $R_e$ for a given value of $R_c$.

To the best of our knowledge, at present there exist no studies that solve the above problem, even for the
simplest of cases. In this appendix, we prove that the alphabet size of $U$ can be reduced from $5$ letters to just
$3$ letters without any loss of optimality (Proposition \ref{prop3letters}), and then provide arguments to back
our Conjecture \ref{conj2letters} that the boundary of the $(R_e,R_c)$ rate region is achieved by a binary auxiliary
random variable $U$. Finally, we prove Proposition \ref{propbsc} which states that if $U$ is considered binary, then
the optimal auxiliary channel, which connects $U$ to the channel input random variable $X$, is a BSC. 

Our model consists of a main channel and an eavesdropper channel modeled as BSCs with crossover probabilities $\epsilon$ and $\delta$,
respectively, such that $\epsilon<\delta$ and $\epsilon, \delta \in [0,0.5]$.

\subsection{Proof of Proposition \ref{prop3letters}}

Following the proof of the admissibility of the size constraints in \cite{csisz1}, we make the following denotations:
\be\label{5quant111}
f_x(\mathbf{p})=Pr(X=0|\mathbf{p})=\mathbf{p}(0)=p,
\ee
\be
f_y(\mathbf{p})=H(Y|\mathbf{p})=h(\epsilon+p-2\epsilon p),
\ee
\be
f_z(\mathbf{p})=H(Z|\mathbf{p})=h(\delta+p-2\delta p),
\ee
where $\mathbf{p}$ denotes the probability mass function (p.m.f.) of $X$, while $f_y(\mathbf{p})$ and $f_z(\mathbf{p})$ are
the respective entropies of $Y$ and $Z$, when $X$ has the p.m.f. given by $\mathbf{p}$.
In the remainder of this appendix we shall denote $a\to b=a+b-2ab$, as the formula is the same as the
crossover probability of a concatenation of two BSCs with respective crossover probabilities $a$ and $b$.

Think of $\mathbf{p}$ as a function under the control of the random variable $U$. Thus, for any $u$ in the alphabet of $U$,
if $U=u$, then the p.m.f. of $X$ becomes $\mathbf{p}_u$.
We can now write, as in \cite{csisz1},
\be\label{5quant1}
Pr(X=0)=\sum_{u}Pr(U=u)f_x(\mathbf{p}_u), 
\ee
\be\label{5quant2}
I(U;Z)=H(Z)-H(Z|U)=\nonumber\\
=H(Z)-\sum_{u}Pr(U=u)f_z(\mathbf{p}_u),
\ee
\be\label{5quant3}
I(X;Y|U)=H(Y|U)-H(Y|X)=\nonumber\\
=\sum_{u}Pr(U=u)\left[f_y(\mathbf{p}_u)-h(\epsilon)\right],
\ee
and
\be\label{5quant4}
I(X;Z|U)=H(Z|U)-H(Z|X)=\nonumber\\
=\sum_{u}Pr(U=u)\left[f_z(\mathbf{p}_u)-h(\delta)\right],
\ee
where we used the fact that $U\to X\to YZ$ form a Markov chain and that $H(Z|X)$ and $H(Y|X)$
are independent of the actual probability distribution of $X$ (the variables are related through BSCs).
Note that $H(Z)$ is completely determined by the channel coefficients $\epsilon$ and $\delta$ and by
$Pr(X=0)$ defined in (\ref{5quant1}).

Consider the triple $(f_x(\mathbf{p}),f_y(\mathbf{p}),f_z(\mathbf{p}))=(p,h(\epsilon+p-2\epsilon p),h(\delta+p-2\delta p))$
and note that all of the quantities in (\ref{5quant1}) - (\ref{5quant4}) above are expressed in terms of the same convex combination of
one of the members of our triple. In other words, any set of feasible values for the quantities in (\ref{5quant1}) - (\ref{5quant4})
is uniquely determined by a point in the convex hull of the set
$\mathscr{C}=\{(p,h(\epsilon+p-2\epsilon p),h(\delta+p-2\delta p))|p\in [0,0.5]\}$, which is a 3D space curve.
Note here that for any $p\in [0.5,1]$ we can find a $p'\in [0,0.5]$ that yields the same values for
$I(U;Z)$, $I(X;Y|U)$ and $I(X;Z|U)$.

By Caratheodory's theorem, since $\mathscr{C}\subset \mathbb{R}^3$, any point in the convex hull of $\mathscr{C}$
can be expressed as a convex combination of only four points belonging to $\mathscr{C}$.
Using the same strengthened version of Caratheodory's theorem, due to Eggleston, as in \cite{csisz1}, we can state
that, since $\mathscr{C}$ is a \emph{connected}\footnote{See definitions in \cite{munkres}. A separation of a
topological space $\mathcal{S}$ is a pair of \emph{nonempty, disjoint, open} subsets of $\mathcal{S}$, whose union is
$\mathcal{S}$. The space $\mathcal{S}$ is \emph{connected} if there does not exist a separation of $\mathcal{S}$.}
subset of $\mathbb{R}^3$, any point in its convex hull
can be expressed as a convex combination of only three points belonging to $\mathscr{C}$ (Theorem 18 (ii) on page 35
of \cite{eggleston1}).
This implies that it is enough to consider only three values of $p$ to be able to produce
any triple of feasible values for the quantities in (\ref{5quant3}) - (\ref{5quant4}). But since $p$ is controlled by
the value of the auxiliary random variable $U$, we can therefore let $U$ be ternary.$\Box$

\subsection{Arguments Supporting Conjecture \ref{conj2letters}}

In Conjecture \ref{5lemma2ponts} below we state that, due to the special form of the set $\mathscr{C}$ defined in the
previous subsection, we
can actually express any point in its convex hull as the convex combination of only two of its points. 

This would imply that it is enough to consider only two values of $p$ to be able to produce
any triple of feasible values for the quantities in (\ref{5quant3}) - (\ref{5quant4}) and hence we can let $U$ be binary.

\begin{conj}\label{5lemma2ponts}
Consider the 3D space curve given by
$\mathscr{C}=\{(p,h(\epsilon+p-2\epsilon p),h(\delta+p-2\delta p))|p\in [0,0.5]\}$. Any point in the
convex hull of $\mathscr{C}$ can be expressed as the convex combination of only two points belonging to
$\mathscr{C}$.
\end{conj}

\emph{Supporting arguments}
\vspace*{4pt}

Recall the denotation $x\to p =x+p-2xp$.
The space curve $\mathscr{C}$, along with its projections onto the $(p,h(\epsilon\to p))$ and $(p,h(\delta\to p))$ planes,
is represented in Figure \ref{5figappspc1}. We shall henceforth call the $p$ axis the ``abscissa'' axis, because it is the
common abscissa axis of both $(p,h(\epsilon\to p))$ and $(p,h(\delta\to p))$ planes.
Also represented in the figure is a random point $M$ in the convex hull of $\mathscr{C}$,
which was obtained as the convex combination of three points $A,B$ and $C$ belonging to $\mathscr{C}$. Due to Eggleston's extension of
Caratheodory's theorem \cite{eggleston1}, we know that any point in the convex hull of $\mathscr{C}$ can be obtained in this manner.
In the remainder of this argument we shall denote by
$P_d$ the projection of the point $P$ onto the $(p,h(\delta\to p))$ plane, and by $P_e$ the projection of the point $P$
onto the $(p,h(\epsilon\to p))$ plane, for any point $P\in\{A,B,C,D,E,F,G,M,X,Y\}$. Moreover, we denote by $\mathscr{C}_d$ and
$\mathscr{C}_e$ the projections of the space curve $\mathscr{C}$ on the two planes, respectively.

The present conjecture shows that in fact the point $M$ can be obtained as the convex combination of only two points of $\mathscr{C}$ - in
Figure \ref{5figappspc1} these points were denoted by $X$ and $Y$. 
\begin{figure}[]
\centering
\includegraphics[scale=0.47]{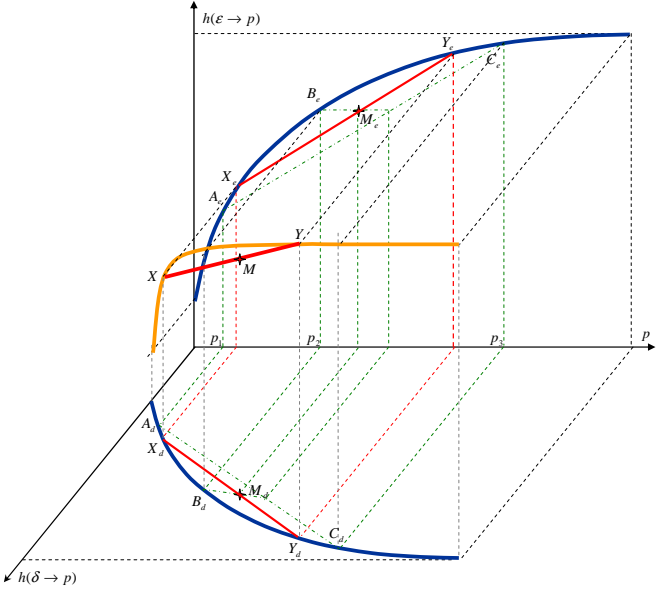}
\caption{The space curve and its projections onto the $(p,h(\epsilon\to p))$ and $(p,h(\delta\to p))$ planes.}\label{5figappspc1}
\end{figure}

This is equivalent to showing that there exist two values $p_x$ and $p_y$ of $p$, such that if we denote the points
$X_e=(p_x,h(\epsilon\to p_x))$, $X_d=(p_x,h(\delta\to p_x))$, $Y_e=(p_y,h(\epsilon\to p_y))$ and $Y_d=(p_y,h(\delta\to p_y))$,
then $M_e$ belongs to the line segment connecting $X_e$ and $Y_e$, and simultaneously $M_d$ belongs to the line segment
connecting $X_d$ and $Y_d$.
At this point, assume that the following remark is true.

\begin{rem}\label{5remappspc1}
(This remark has been checked numerically. However, we currently do not have a theoretical proof.)
Consider four random points $A,D,B,C$ on the space curve $\mathscr{C}$, such that their respective abscissae
$p_1,p_4,p_2,p_3$ satisfy $p_1<p_4<p_2<p_3$, and construct their projections $A_e,D_e,B_e,C_e$ and $A_d,D_d,B_d,C_d$
on the planes $(p,h(\epsilon\to p))$ and $(p,h(\delta\to p))$, respectively.
Then the abscissa of the intersection of the segments $A_eB_e$ and $D_eC_e$ is greater than the abscissa of the
intersection of the segments $A_dB_d$ and $D_dC_d$.
The result is illustrated in Figure \ref{5figappspc2} for two tuples of points, namely $(A,D,B,C)$ and $(A,B,F,C)$.
\end{rem}

\begin{figure}[]
\centering
\includegraphics[scale=0.47]{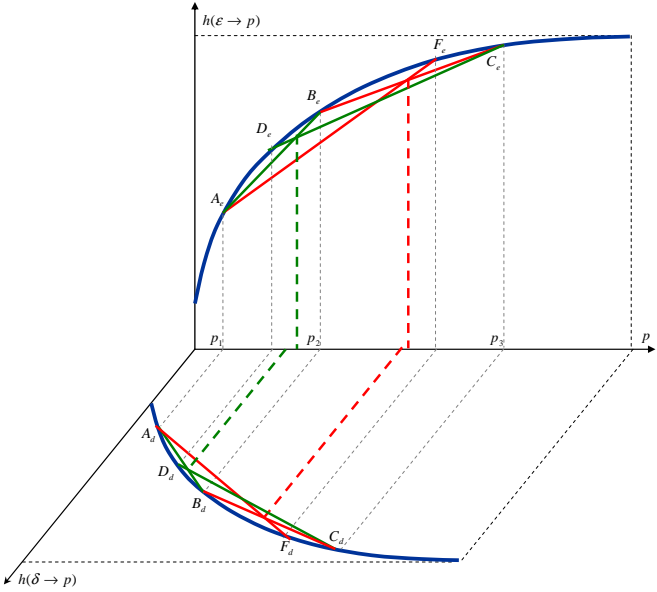}
\caption{Projections of the space curve: simplified problem.}\label{5figappspc2}
\end{figure}

Recall that the points $A,B$ and $C$ determine our point of interest $M$, that is $M=aA+bB+cC$, where
$a,b,c\in[0,1]$ and $a+b+c=1$. This implies that the intersection between the segments $A_eB_e$ and $C_eM_e$, and the
intersection between $A_dB_d$ and $C_dM_d$ have the same abscissa, namely $\frac{ap_1+bp_2}{a+b}$. Due to Remark
\ref{5remappspc1} above, this means that the segment $C_eM_e$ intersects the curve $\mathscr{C}_e$ at a point $E_e$ which has
an abscissa $p_{1,e}$ which is less than the abscissa $p_{1,d}$ of the intersection $D_d$ between $C_dM_d$ and $\mathscr{C}_d$,
as illustrated in Figure \ref{5figappspc3}.

\begin{figure}[]
\centering
\includegraphics[scale=0.47]{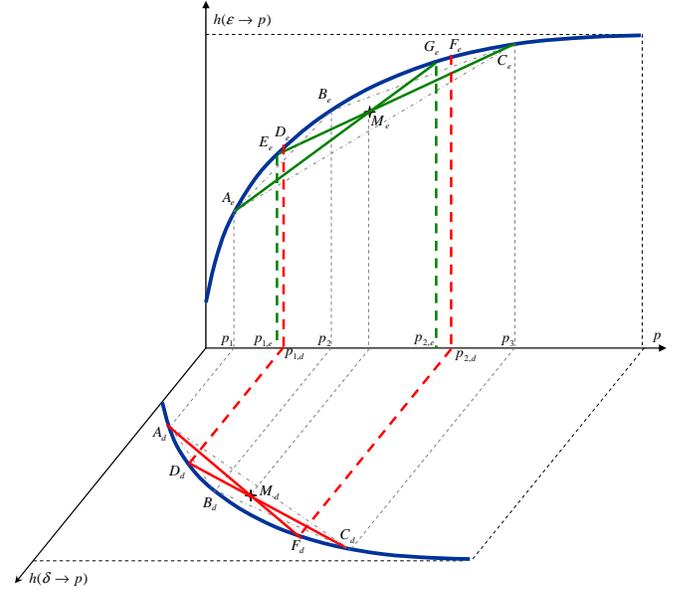}
\caption{Projections of the space curve: existence of a solution.}\label{5figappspc3}
\end{figure}

Denote by $D_e$ the point of $\mathscr{C}_e$ with the same abscissa $p_{1,d}$ as $D_d$. It is clear that the segment
$D_eC_e$ passes above the point $M_e$, while $D_dC_d$ passes through $M_d$.

By a similar rationale, the intersection between the segments $A_eM_e$ and $B_eC_e$, and the
intersection between $A_dM_d$ and $B_dC_d$ have the same abscissa, namely $\frac{bp_2+cp_3}{b+c}$. Due to Remark
\ref{5remappspc1}, this means that the segment $A_eM_e$ intersects the curve $\mathscr{C}_e$ at a point $G_e$ which has
an abscissa $p_{2,e}$ which is less than the abscissa $p_{2,d}$ of the intersection $F_d$ between $A_dM_d$ and $\mathscr{C}_d$ (see
Figure \ref{5figappspc3}). Denote by $F_e$ the point of $\mathscr{C}_e$ with the same abscissa $p_{2,d}$ as $F_d$.
It is clear that the segment $A_eF_e$ passes below the point $M_e$, while $A_dF_d$ passes through $M_d$.

This implies that there exists a value $p_x\in [p_1,p_{1,d}]$ of $p$ such that, if we denote
$X_e=(p_x,h(\epsilon\to p_x))$ and $X_d=(p_x,h(\delta\to p_x))$, then the segments $X_eM_e$ and $X_dM_d$ intersect
the curves $\mathscr{C}_e$ and $\mathscr{C}_d$, respectively, at points $Y_e$ and $Y_d$ with the same abscissa $p_y\in[p_{2,d},p_3]$.
Hence $X_e$ and $X_d$ are the projections of a point $X\in \mathscr{C}$, and $Y_e$ and $Y_d$ are the projections
of a point $Y\in \mathscr{C}$, and the segment $XY$ goes through $M$.$\Box$

\subsection{Proof of Proposition \ref{propbsc}}

\begin{figure*}[t!]
\normalsize \newcounter{mytempeqncnt}
\setcounter{mytempeqncnt}{\value{equation}}
\setcounter{equation}{65}
\be\label{5exprg2}
g''(x)=\frac{q}{x(1-x)+\mu(\alpha)}+\frac{1-q}{x(1-x)+\mu(\beta)}-\frac{1}{x(1-x)+\mu(\gamma)}=\nonumber\\
=\frac{x(1-x)[\mu(\gamma)-q\mu(\alpha)-(1-q)\mu(\beta)]+\mu(\gamma)(q\mu(\beta)+(1-q)\mu(\alpha))-\mu(\alpha)\mu(\beta)}
{(x(1-x)+\mu(\alpha))(x(1-x)+\mu(\beta))(x(1-x)+\mu(\gamma))}
\ee
\setcounter{equation}{\value{mytempeqncnt}}
\hrulefill
\vspace*{4pt}
\end{figure*}
\addtocounter{equation}{1}

Let $U$ belong to $\{0,1\}$ (whether that is optimal or not is still an open problem),
and denote $q=Pr(U=0)$. Since $X$ is also binary, the channel between $U$ and $X$ can be completely characterized by
two transition probabilities. Denote $\alpha=Pr(X=1|U=0)$ (this implies $Pr(X=0|U=0)=1-\alpha$), and
$\beta=Pr(X=0|U=1)$ (this implies $Pr(X=1|U=1)=1-\beta$).

Note that (\ref{5rel11111}) and (\ref{5rel22222}) can be rewritten as:
\be
R_e\leq [H(Z|X)-H(Y|X)]-\nonumber\\
-[q(H(Z|U=0)-H(Y|U=0))+\nonumber\\+(1-q)(H(Z|U=1)-H(Y|U=1))]
\ee
and
\be
R_c\leq H(Z)-[qH(Z|U=0)+(1-q)H(Z|U=1)],
\ee
With the notation above, the upper bounds can be written as
\be
R_{e,u}(q,\alpha,\beta)\leq [h(\delta)-h(\epsilon)]-\nonumber\\
-[q(h(\alpha\to\delta)-h(\alpha\to\epsilon))+\nonumber\\
+(1-q)(h(\beta\to\delta)-h(\beta\to\epsilon))]
\ee
and
\be
R_{c,u}(q,\alpha,\beta)\leq h(q\alpha+(1-q)(1-\beta)\to \delta)- \nonumber\\
-[qh(\alpha\to\delta)+(1-q)h(\beta\to\delta)],
\ee
where $a\to b$ stands for $a(1-b)+b(1-a)=a+b-2ab$ as before, and we emphasized the dependence of the upper bounds upon
the triple $(q,\alpha,\beta)$.

In what follows we take a contradictory approach. Consider any triple $(q,\alpha,\beta)$ and denote
\be
R_x(q,\alpha,\beta)=1-[qh(\alpha\to\delta)+(1-q)h(\beta\to\delta)].
\ee
We show that if we replace this triple by the triple $(0.5,\gamma,\gamma)$ (corresponding to a
uniform distribution of $U$ over $\{0,1\}$ and a BSC between $U$ and $X$), such that
\be\label{5eqRxdelta}
R_x(q,\alpha,\beta)=R_x(0.5,\gamma,\gamma)
\ee
(we also prove that such a $\gamma$ exists always), we have
$R_{e,u}(q,\alpha,\beta)\leq R_{e,u}(0.5,\gamma,\gamma)$ and $R_{c,u}(q,\alpha,\beta)\leq R_{c,u}(0.5,\gamma,\gamma)$.
Therefore, a triple $(q,\alpha,\beta)$ for which either $q\neq 0.5$ or $\alpha\neq \beta$ holds cannot be optimal, and
hence the last part of our theorem is proved.

Note that $R_x(q,\alpha,\beta)=R_x(0.5,\gamma,\gamma)$ translates to
\be\label{5eqfinggamma}
qh(\alpha\to\delta)+(1-q)h(\beta\to\delta)=h(\gamma\to\delta),
\ee
Since $qh(\alpha\to\delta)+(1-q)h(\beta\to\delta)\in [0,1]$, the binary entropy function is a bijection over
$[0,0.5]$ and $f(\gamma)=\gamma\to\delta$ with $\delta\in (0,0.5)$ is also a bijection over $[0,0.5]$,
we can always find a $\gamma$ that satisfies (\ref{5eqfinggamma}).
Since $h([q\alpha+(1-q)(1-\beta)]\to \delta)\leq 1$ and $h([0.5\gamma+0.5(1-\gamma)]\to \delta)= h(0.5\to \delta)=0$
it is straightforward to see that

\be
R_{c,u}(q,\alpha,\beta)\leq R_x(q,\alpha,\beta) =\nonumber\\
=R_x(0.5,\gamma,\gamma)=R_{c,u}(0.5,\gamma,\gamma).
\ee
We can now write
\be
R_{e,u}(0.5,\gamma,\gamma)-R_{e,u}(q,\alpha,\beta)=\nonumber\\
=h(\gamma\to\epsilon)-qh(\alpha\to\epsilon)+(1-q)h(\beta\to\epsilon).
\ee

Define $g(x)=h(\gamma\to x)-qh(\alpha\to x)+(1-q)h(\beta\to x)$. From (\ref{5eqRxdelta}) we have that
$g(\delta)=0$, and it is straightforward to see that $g(0.5)=0$. Since we only discuss the case when
$\delta<0.5$, we now know that $g(x)$ has two different zeros over the interval $[0,0.5]$. We need to
show that for any $\epsilon<\delta$ we have $g(\epsilon)>0$.

Denote $g'(x)=\frac{dg(x)}{dx}$ and $g''(x)=\frac{d^2g(x)}{dx^2}$ the first and second order derivatives of $g$.
With the notation $\mu(x)=\frac{x(1-x)}{(1-2x)^2}$, we can write $g''$ as in (\ref{5exprg2}).

Since the denominator of $g''$ is always positive, the equation $g''(x)=0$ reduces to a second degree equation in $x$.
Thus $g''$ has at most two real zeros, which are symmetric with respect to the point $0.5$, and hence at most one zero
(denote it by $z''$)
in the interval $[0,0.5]$. Moreover, since $\mu(x)$ is a strictly convex function of $x$, the coefficient
$-[\mu(\gamma)-q\mu(\alpha)-(1-q)\mu(\beta)]$ of $x^2$ in the numerator of $g''$ is strictly positive. This implies that
$g''(x)>0$ for $x\in[0,z'']$.

Now suppose that $g(x)$ had more than two zeros on the interval $[0,0.5]$. Then $g'(x)$ would have at least two zeros
on the open interval $(0,0.5)$, and hence a total of three zeros in $[0,0.5]$ (it is straightforward to check that
$g'(0.5)=0$). Thus $g''$ would need to have at least two zeros in $(0,0.5)$. But we have already shown that this is impossible.
Therefore, $g(x)$ has only two zeros in the interval $[0,0.5]$ (these are $\delta$ and $0.5$).

As a consequence, $g'$ has at least one zero in $(\delta,0.5)$ -- denote this zero by $z'$. Since $g'$ has a zero in $0.5$,
this implies that the zero $z''$ of $g''$ is in the interval $(z',0.5)$. We can now write $\delta<z'<z''$.
We already know that $g''(x)>0$ on $[0,z'')$, thus $g'(x)$ is strictly increasing on $[0,z']$, and since
$g'(z')=0$, this means that $g'(x)<0$ on $[0,\delta]$. But since $g(\delta)=0$, this means that
for any $\epsilon<\delta$ we have $g(\epsilon)>0$.
Our argument is now complete.$\Box$


\bibliographystyle{IEEEtran}
\bibliography{jamming}
\end{document}